%% file: main.tex
\documentclass{llncs}
\usepackage{main}

\begin{document}

\title{Certified Verification of Relational Properties\thanks{
Part of this work was funded by the AESC project supported by the Ministry of
Science, Research and Arts Baden-W\"{u}rttemberg (Ref: 33-7533.-9-10/20/1).}
}

% REAL AUTHORS
\author{
	\mbox{}
%	\hspace{-4mm}
%ORCID is given in upper index for each author when available	
	Lionel Blatter\inst{1} \and
	Nikolai Kosmatov\inst{2,3}$^{(0000-0003-1557-2813)}$ \and\\ 
	Virgile Prevosto\inst{2}$^{(0000-0002-7203-0968)}$ \and\\
	Pascale Le Gall\inst{4}$^{(0000-0002-8955-6835)}$}
\institute{
	Karlsruhe Institute of Technology \\
	\email{firstname.lastname@kit.edu}
	\and
	Université Paris-Saclay, CEA, List, 91120, Palaiseau, France\\
	\email{firstname.lastname@cea.fr}
	\and    
	Thales Research \& Technology, 91120, Palaiseau, France %\\
%	\email{nikolaikosmatov@gmail.com}	
	\and
	CentraleSupélec, Université Paris-Saclay, 91190 Gif-sur-Yvette France\\
	\email{firstname.lastname@centralesupelec.fr}
}

\maketitle
\vspace{-3mm}
\begin{abstract}
  The use of function contracts to specify the behavior of functions
  often remains limited to the scope of a single function call.
  \emph{Relational properties} link several function calls together
  within a single specification. They can express more advanced properties of a
  given function, such as non-interference, continuity, or monotonicity.
  They can also relate calls to different functions, for instance, to show that an
  optimized implementation is equivalent to its original counterpart.
  However, relational properties cannot be expressed and
  verified directly in the
  traditional setting of modular deductive verification.
  Self-composition has been proposed to overcome this
  limitation, but  
%  the simplicity of its principle
%  comes at the cost of
  it requires
  complex transformations and additional separation hypotheses
  for real-life languages with pointers.
  %\commentNK{TBD} 
%  for non-trivial examples.
  We propose %in this paper 
  %an alternative 
  a novel approach
  that is not based on code transformation and avoids those drawbacks. %Instead,
  %we %show how to 
  It directly applies a verification condition generator
  to produce logical formulas that must be verified
  to ensure a given relational property.
  The approach has been fully formalized and proved sound in the \coq proof assistant.

%\medskip
%\textbf{Keywords:}
%relational properties,
%formal specification,
%deductive verification,
%verification condition generator,
%\coq
\end{abstract}

%-------------------------------------------------------------
\vspace{-5mm}
\section{Introduction}
\vspace{-2mm}
\label{sec:intro}
\input{introduction.tex}

% -----------------------------------------------------------

%\section{Background}
\vspace{-2mm}
\section{Syntax and Semantics of the \langname Language}
\vspace{-2mm}
\label{sec:background}
\label{sec:considered-language}
\input{background.tex}

%-------------------------------------------------------------
\vspace{-2mm}
\section{Functional Correctness}
\vspace{-2mm}
\label{sec:funct-corr}
\input{funct-corr.tex}

%-------------------------------------------------------------
\vspace{-2mm}
\section{Relational Properties}
\vspace{-2mm}
\label{sec:rel-prop}
\input{rel-prop.tex}

%-------------------------------------------------------------
\vspace{-2mm}
\section{Verification Condition Generation for Hoare Triples}
\vspace{-2mm}
\label{sec:verif-cond-gener}
\input{vcg.tex}

% ----------------------------
\vspace{-2mm}
\section{Verification of Relational Properties}
\vspace{-2mm}
\label{sec:relat-prop-verif}
\input{verifrelaprop.tex}
\vspace{-2mm}
\section{Related Work}
\vspace{-2mm}
\label{sec:related}
\input{related.tex}
\vspace{-3mm}
\section{Conclusion} % and Perspectives}
\vspace{-2mm}
\label{sec:conclusion}
\input{conclusion.tex}

\bibliographystyle{splncs04}
\bibliography{biblio}

%\end{document}
\clearpage
\appendix

\begin{center}
  {\huge \bf Appendix}
\end{center}
\input{appendix.tex}

\end{document}

%% file: introduction.tex
Modular deductive verification~\cite{DBLP:journals/cacm/Hoare69}
allows the user to prove that a function respects its formal specification.
More precisely,
for a given function $f$, any individual call to $f$ can be proved to respect the \emph{contract} of $f$,
that is, basically an implication: if the given \emph{precondition} is true before the call and the call
terminates\footnote{Termination can be assumed (partial correctness) or proved separately (full correctness) in a well-known way~\cite{Floyd1967}; 
for the purpose of this paper we can assume it.},
the given \emph{postcondition} is true after it.
However, some kinds of properties
are not easily reducible to a single function call.
Indeed, it is frequently necessary to express a property that involves several functions
or relates the results of several calls to the same function for different arguments.
Such properties are known as \emph{relational properties}~\cite{DBLP:conf/popl/Benton04}.

Examples of such relational properties include monotonicity
({i.e.} $x\leq y\Rightarrow f(x) \leq f(y)$), involving two calls,
or transitivity
($cmp(x,y) \geq 0 \wedge cmp(y,z) \geq 0 \Rightarrow cmp(x,z)\geq 0$),
involving three calls.
In secure information flow~\cite{DBLP:journals/mscs/BartheDR11},
\emph{non-interference} is also a relational property. Namely,
given a partition of program variables between high-security variables
and low-security variables,
a program is said to be non-interferent if %and only if
any two executions starting from states in which the low-security
variables have the same initial values will end up
in a final state where the low-security variables have the same values.
In other words, high-security variables cannot interfere with low-security ones.

Relational properties can also relate calls to different functions.
%,
%for instance, to show that the encryption of a message followed by 
%its decryption with an appropriate key gives back the original message.
For instance, in the verification of voting rules~\cite{BeckertBormerEA2016},
relational properties are used for defining specific properties (such as monotonicity, anonymity or consistency).
Notably, applying the voting rule to a sequence of ballots and a permutation of the same sequence of ballots
must lead to the same result, %\emph
{i.e.} the order in which the ballots are passed to the
voting function should not have any impact on the outcome.

\begin{figure}[tb]
\centering
\begin{tabular}{c|c|l}
	
\begin{minipage}{3.3cm}
\begin{lstlisting}[stepnumber=0,mathescape]
//$\mbox{\rm C program }\CprogI:$
x3  = *x1;
*x1 = *x2;
*x2 = x3;
//$\mbox{\rm C program }\CprogII:$
*x1 = *x1 + *x2;
*x2 = *x1 - *x2;
*x1 = *x1 - *x2;
\end{lstlisting}	
\end{minipage}

&\ 

\begin{minipage}{4.5cm}
\begin{lstlisting}[stepnumber=0,mathescape]
//$\mbox{\rm Composed C program }\CprogIII:$
x3_1  = *x1_1;
*x1_1 = *x2_1;
*x2_1 = x3_1;

*x1_2 = *x1_2 + *x2_2;
*x2_2 = *x1_2 - *x2_2;
*x1_2 = *x1_2 - *x2_2;
\end{lstlisting}	
\end{minipage}

&\ 

\begin{minipage}{4cm}
\vspace{3mm}
\normalsize
$%\begin{equation*}
\cswI\triangleq
\begin{array}{l@{}l}
\locval_3 &:= *\locval_1;\\
*\locval_1 &:= *\locval_2;\\
*\locval_2 &:= \locval_3;
\end{array}
$%\end{equation*}	
\vspace{3mm}
\\
$%\begin{equation*}
%\qquad\mathrm{and}\qquad
\cswII\triangleq
\begin{array}{l@{}l}
*\locval_1 &:= *\locval_1+ *\locval_2;\\
*\locval_2 &:= *\locval_1 - *\locval_2;\\
*\locval_1 &:= *\locval_1 - *\locval_2
\end{array}
$%\end{equation*}	
\end{minipage}

\end{tabular}

%
%\begin{minipage}{12cm}
%	\normalsize
%%\begin{equation*}
%%c_1\triangleq
%%  \begin{array}{l@{}l}
%%    \locval_3 &:= *\locval_1;\\
%%    *\locval_1 &:= *\locval_2;\\
%%    *\locval_2 &:= \locval_3;
%%    \end{array}
%%\qquad\mathrm{and}\qquad
%%c_2\triangleq
%%  \begin{array}{l@{}l}
%%    *\locval_1 &:= *\locval_1+ *\locval_2;\\
%%    *\locval_2 &:= *\locval_1 - *\locval_2;\\
%%    *\locval_1 &:= *\locval_1 - *\locval_2
%%  \end{array}
%%\end{equation*}
%\begin{equation*}
%  \,\,\left\{\rela{P}\right\}\,\,
%  \begin{array}{l@{}l}
%    \ben{
%    \begin{array}{l@{}l}
%      \locval_{31} &:= *\locval_{11};\\
%      *\locval_{11} &:= *\locval_{21};\\
%      *\locval_{21} &:= \locval_{31};
%    \end{array}
%    }{1}\\
%    ~\\
%    \ben{
%    \begin{array}{l@{}l}
%      *\locval_{12} &:= *\locval_{12}+ *\locval_{21};\\
%      *\locval_{22} &:= *\locval_{12} - *\locval_{22};\\
%      *\locval_{12} &:= *\locval_{12} - *\locval_{22}
%    \end{array}
%     }{2}
%  \end{array}
%   \,\,\left\{\rela{Q}\right\}.
%\end{equation*}
%\end{minipage}
\vspace{-1mm}
\caption{Two C programs $\CprogI$ and $\CprogII$ swapping \lstinline'*x1' and 
\lstinline'*x2', their composition $\CprogIII$, and their counterparts $\cswI$ and $\cswII$ in language $\mathcal{L}$ (defined below).}
%\vspace{-5mm}
\label{fig:ex-rel-prop-swaps-self}
\end{figure}

%\commentNK{Motivating example in progress by NK}

\vspace{-2mm}
\paragraph{Motivation.}
Lack of support for relational properties in verification tools
was already faced by industrial users (e.g. in~\cite{BishopBC13} for C programs).
The usual way to deal with this limitation is to use
\emph{self-composition}~\cite{DBLP:journals/mscs/BartheDR11,DBLP:conf/fm/SchebenS14,blatterKGP17},
product program~\cite{DBLP:conf/fm/BartheCK11} or
other self-composition optimizations~\cite{ShemerCAV2021}.
Those techniques are based on code transformations
that are relatively tedious and error-prone. Moreover, 
they are hardly applicable in practice to real-life programs 
with pointers like in C.
%Generally speaking, 
Namely, self-composition requires that the compared executions
%being compared
operate on completely {separated} (i.e. disjoint) memory areas, which might
be extremely difficult to ensure for complex programs with pointers.

%\commentLBi{and Related work}
\begin{example}[Motivating Example]
\label{ex:motiv-ex}
%NK: "Figure" MUST be non abbreviated when FIRST word of a sentence
Figure~\ref{fig:ex-rel-prop-swaps-self} shows an example of two simple C programs
%$\Cprog_1,\Cprog_2$ 
performing
a swap of the values %  \lstinline'*x1' and \lstinline'*x2'
referred to by pointers \lstinline'x1' and \lstinline'x2' (of type \lstinline'int*').
Program $\CprogI$ 
uses an auxiliary 
%local 
variable \lstinline'x3' (of type \lstinline'int'), while 
$\CprogII$ performs an in-place swap using arithmetic operations.
As usual in that case, to work correctly, 
each of these programs needs some separation hypotheses:
pointers \lstinline'x1' and \lstinline'x2' should be \emph{separated}
(that is, point to disjoint memory locations) and must not point
to \lstinline'x1', \lstinline'x2' themselves and, for $\CprogI$, to \lstinline'x3'.

%In addition to the aforementioned---well-justified---conditions, These restrictions come from the 

Consider a relational property, denoted $\Rsw$, stating that both programs,
executed from two states 
%NK: I propose this statement to avoid a misunderstanding that  *x1 == *x2 must hold for each program
in which each of \lstinline'*x1' and \lstinline'*x2' has the same value, will end up in two states also
having the same values in these locations. 
To prove this relational property using self-composition, 
one typically has to generate a new C program $\CprogIII$ (see Fig.~\ref{fig:ex-rel-prop-swaps-self}) composing $\CprogI$ and $\CprogII$.
To avoid name conflicts, we rename their variables  
by adding, resp., suffixes ``\lstinline'_1''' and  ``\lstinline'_2'''.
The relational property  $\Rsw$ is then expressed by a 
contract 
%
%NK: as Hoare triples are not yet defined, I think it is better to avoid them here
%
%(or a so-called Hoare triple $\{ P \} \CprogIII \{ Q \}$, defined below) 
%defined in Appendix.} 
of $\CprogIII$ with a precondition $P$ and a postcondition $Q$.
Obviously, both $P$ and $Q$ must include the equalities: 
\lstinline'*x1_1==*x1_2' and 
\lstinline'*x2_1==*x2_2', and 
$P$ must also require the aforementioned separation hypotheses 
necessary for each function. But for programs with pointers and aliasing, this is 
not sufficient: the user also has to specify additional separation 
hypotheses\footnote{For convenience of the reader, $P$ and $Q$ are defined in detail in Appendix~\ref{app-motiv-ex}.}
between variables coming from the different programs, that is, in our example,  that each of 
\lstinline'x1_1' and \lstinline'x2_1' is separated from each of \lstinline'x1_2' and
\lstinline'x2_2'. Without such hypotheses,
a deductive verification tool cannot show, for example, that a modification of 
\lstinline'*x1_1' does not impact \lstinline'*x1_2' in the composed program $\CprogIII$, 
and is thus unable to deduce the required property. 
For real-life programs, such separation hypotheses 
can be hard to specify or generate. It can become even more complicated for programs
with double or multiple indirections.
%
%NK: Finally even a partial definition is too long, I prefer to keep it in Appendix
%
%\begin{center}
%	\scriptsize
%$P:$ \, \lstinline[basicstyle=\scriptsize\ttfamily]'*x1_1==*x1_2 && *x2_1==*x2_2 && x1_1!=x1_2 && x1_1!=x2_2 && x1_1!=x2_2 && x1_1!=x2_2 && x1_1!=x2_2 ',\\
%$Q:$ \quad \lstinline[basicstyle=\scriptsize\ttfamily]'*x1_1==*x1_2 && *x2_1==*x2_2',
%\end{center}
%(We use C and \acsl notation for C programs.) 
\qed
\end{example}

%\commentLBi{
%%  Firgure~\ref{fig:ex-rel-prop-swaps-self}
%%  show an example of two different implementation of a swap function.
%  By apply self-composition we notice that we require not only
%  separation of pointers in $\langle1\rangle$
%  and $\langle2\rangle$ but also between $\langle1\rangle$ and $\langle2\rangle$.
%}

%\commentNK{NK questions: transformations vs. separation, link to product programs...}

\vspace{-4mm}
\paragraph{Approach.}
This paper proposes an alternative approach %to self-composition
that is not based on code transformation or relational rules.
%Instead, we show how to 
It directly uses a verification condition  generator (VCGen)
to produce logical formulas to be verified (typically, with an
automated prover)
to ensure a given relational property. 
It requires no extra code processing 
(such as sequential composition of programs or variable renaming).
Moreover, no additional separation hypotheses---in addition to those that are anyway needed for each function to work---are required. 
The locations of each program are separated by construction: 
each program  
%command $\com_i$ 
has its own memory state.
% $\mem_i$.
The language $\langname$ considered in this work 
was chosen as a minimal language representative of the 
main issues relevant for relational property verification: 
it is a standard \textsc{While} language enriched %xtended 
with annotations, procedures and pointers (see 
programs $\cswI$ and $\cswII$ in Fig.~\ref{fig:ex-rel-prop-swaps-self} for examples; we use a lower-case letter $\com$ for $\langname$ programs and a capital letter $\mathcal{C}$ for C programs).
Notably, the presence of dereferencing  and address-of operations makes
it representative of various aliasing problems 
with (possibly, multiple) pointer dereferences of a real-life
language like C.
%The proposed semantics can be compared to the one of an
%assembly language, with low-level pointer arithmetic and procedure calls.
%%In particular, 
%Currently, the language does not allow modeling
%arbitrary mappings of variables
%to physical memory and dynamic memory allocation, 
%like in a real-life language like C.
%Instead, it fixes a specific mapping (where $\locval_i$
%lies at address $i$).
%Considering a simple \textsc{While}-like 
%programming language with pointers and procedures,
We formalize the proposed approach and prove\footnote{The \coq development is
%available 
at 
%\url{https://bit.ly/3FJhE41}
\url{https://github.com/lyonel2017/Relational-Spec}, where the version corresponding to this paper is tagged iFM2022.}
its soundness in the \coq proof assistant~\cite{Coq}.
Our \coq development contains about 3400 lines.

%\paragraph{Advantages of the Proposed Method.}
%\commentNK{I think I will integrate this paragraph into intro or motivating example}
%
%An important benefit of the
%proposed method for relational property verification presented in
%Section~\ref{sec:relat-prop-verif}
%is that it requires no code preprocessing, {i.e.} no additional
%work (such as the sequential composition of programs and variable renaming).
%Moreover, no additional separation hypotheses---in addition to those that are anyway needed for each function to work---
%are required. Indeed, the locations
%of each program are separated by construction: each command $\com_i$ has
%its own memory state $\mem_i$.
%It is well-illustrated by Example~\ref{ex:rel-prop},
%%\ref{ex:rel-prop-vc}, \ref{ex:rel-prop-vc}
%where we only had to express the separation hypotheses
%that are anyway needed for the swap functions to work
%(e.g. the precondition that $\locval_2$ must not point to  $\locval_1,\locval_2,\locval_3$)
%while both programs manipulate the same locations $\locval_1,\locval_2$,
%and even worse,
%the unknown locations pointed by them that can be any locations
%$\locval_i$, $i>3$.
%The proposed method does not require additional frame rules for procedures
%to ensure that locations are left unchanged by other commands involved
%in the relational property.

\vspace{-2mm}
\paragraph{Contributions.}
The contributions of this paper include:
\begin{itemize}
\item a \coq formalization and proof of soundness of recursive Hoare triple verification with a verification condition generator on
a representative language with procedures and aliasing;
\item a novel method for verifying relational properties using a verification condition generator, without relying on code transformation 
(such as
self-composi\-tion) 
or making additional
separation hypotheses in case of aliasing;
\item a \coq formalization and proof of soundness of
the proposed method 
of relational property verification for the considered language.
\end{itemize}

\vspace{-2mm}
\paragraph{Outline.}
%The rest of this paper is structured as follows. 
Section~\ref{sec:background}
introduces an imperative language $\langname$ used in this work.
Functional correctness 
%of a command of $\langname$ 
is defined in Section~\ref{sec:funct-corr},
and relational properties in Section~\ref{sec:rel-prop}. Then, we prove
the soundness of a VCGen %verification condition generator 
%for $\langname$ 
in Section~\ref{sec:verif-cond-gener},
and show how it can be soundly extended to verify relational properties 
%and prove the soundness of this extension 
in Section~\ref{sec:relat-prop-verif}. Finally, we 
present related work 
in Section~\ref{sec:related}
and
concluding remarks in Section~\ref{sec:conclusion}.
%Related work is presented when it is relevant, in various sections
%throughout the paper.

%% file: background.tex
%This section defines the language, denoted
%$\langname$, we use throughout this paper. It is a simple imperative language,
%with procedure calls and dereferences.

%\vspace{-2mm}
%\subsection{Syntax of the Considered Imperative Language}
\subsection{Notation for Locations, States, and Procedure Contracts}
%\vspace{-1mm}
\label{sec:loc-state-contract}
\label{sec:program-grammar}

%\subsubsection{Context.}
%In the following, we use 
We denote by $\Natset=\{0,1,2,\dots\}$ the set of natural numbers,
by $\Natset^*=\{1,2,\dots\}$ the set of nonzero natural numbers,
and by $\Boolset = \{\True,\False\}$ the set of Boolean values.
%Furthermore, 
Let $\Locval$ be the set of program \emph{locations} 
%(or variables) 
and
$\Loccom$ the set of \emph{program (procedure) names}, and let
$\locval,\locval',\locval_1, ...$ and $\loccom,\loccom',\loccom_1, ...$
denote  metavariables ranging over those respective sets.
We assume that there exists a bijective function 
%$\mapping$ from $\Locval$ to
$\Natset\to\Locval$, so that $\Locval=\{\locval_i\,|\,i\in\Natset\}$.
Intuitively, we can see $i$ as the \emph{address} of location $\locval_i$.

Let $\Memvar$ be the set of functions 
%mapping natural numbers to natural numbers 
$\memvar:\Natset \to \Natset$, called \emph{memory states}, and let
$\memvar,\memvar',\memvar_1, ...$ denote metavariables ranging over the set.
A state $\memvar$ maps a location to a value using its address:
location $\locval_i$ has value $\memvar(i).$

%Thanks to $\mapping$, a memory state $\memvar$ can also be seen as mapping
%locations to natural numbers, and we will use the notation $\memvar(x_i)$ to
%denote $\memvar(\mapping(x_i))$.
%Moreover, 
We define the \emph{update} operation of a memory state
$\setmem{\mem}{i}{n}$, also denoted by $\memvar[i/n]$,
as the memory state $\memvar'$ mapping each address to the same value
as $\sigma$, except for $i$,
% (resp. $\mapping(x_i)$), 
bound to $n$.
Formally,
$\setmem{\mem}{i}{n}$ is defined by the following rules:
\begin{eqnarray}
  \forall \memvar \in \Memvar, \locval_i \in \Locval, n \in \Natset, \locval_j \in \Locval.\ 
  i = j \Rightarrow \sigma[i/n] (j) = n,\\
  \forall \memvar \in \Memvar, \locval_i \in \Locval, n \in \Natset, \locval_j \in \Locval.\ 
  i  \neq j \Rightarrow \sigma[i/n] (j) = \memvar(j).
\end{eqnarray}

Let $\Memcom$ be the set of functions 
$\memcom:\Loccom \to \Com$,
called \emph{procedure environments},
mapping program names to commands (defined below),
and let $\memcom,\memcom_1, ...$  denote metavariables ranging over $\Memcom$.
We write $\body{\loccom}{\memcom}$ to refer to $\memcom(y)$,
the commands (or \emph{body}) of procedure $y$  for a given procedure environment $\memcom$.

\emph{Assertions} 
are predicates of arity one, taking as parameter a memory state
and returning an equational first-order logic formula.
Let metavariables $P, Q, ...$ range over the set $\Assertion$ of assertions.
For instance, using $\lambda$-notation, assertion
$P$ assessing that location $x_3$ is bound to $2$
%in the environment at which $P$ is
%evaluated
can be defined by 
$P \triangleq \lambda \memvar. \memvar(3) = 2.$
This form will be more convenient
for relational properties (than {e.g.} $x_3 = 2$)
as it makes explicit
the memory states on which a property is evaluated.

Finally, we define the set $\Memann$ of \emph{contract environments} $\memann: \Loccom  \to \Assertion \times \Assertion$,
and metavariables $\memann,\memann_1,...$ to range over $\Memann$.
More precisely,  
$\memann$ maps a procedure name $y$ to the associated
(procedure) \emph{contract} $\memann(y)=(\pre{\loccom}{\memann},\post{\loccom}{\memann})$, composed of a pre- and a postcondition
for procedure $y$. 
%We denote by $\pre{\loccom}{\memann}$ and $\post{\loccom}{\memann}$, respectively,  the
%precondition and postcondition of procedure $y$  for a given $\memann$.
As usual, a procedure contract will allow us
to specify the behavior of a single procedure call, that is,
if we start executing $\loccom$ in a memory state satisfying $\pre{\loccom}{\memann}$, and the evaluation terminates, the final state satisfies $\post{\loccom}{\memann}$.

\vspace{-3mm}
\subsection{Syntax for Expressions and Commands}
\vspace{-2mm}
\label{sec:exp-commands}
%\commentLBi{Move Figure 1 and 2 to appendix}

Let 
$\Aexp$, $\Bexp$ and $\Com$ denote respectively the sets of arithmetic expressions,
Boolean expressions and commands.
We denote by $\aexp,\aexp_1, ...$;
$\bexp,\bexp_1, ...$ and
$\com,\com_1, ...$ metavariables ranging, respectively, over those sets.
Syntax of arithmetic and Boolean expressions is given in Fig.~\ref{fig:exps-coms}.
Constants are natural numbers or Boolean values. 
Expressions use standard arithmetic, %boolean 
comparison and logic binary operators, denoted respectively $\opa ::= \{+, \times, - \}$, $\opb ::= \{<= , = \}$, $\opl ::= \{\lor, \land\}$.
Since we use natural values, the
subtraction is bounded by 0, as in \coq: if $n'>n$, the result of $n-n'$ is considered to be 0.
Expressions also include locations, possibly with a dereference or address operators.

\begin{figure}[tb]
%\centering
\begin{minipage}{3cm}
  \begin{align*}
    \aexp :&:= \nat & \text{natural const.}\\
           &|\ \locval & \text{location}\\
           &|\ *\locval  & \text{dereference}\\
           &|\ \&\locval & \text{address}\\
           &|\ \aexp_1 \opa \aexp_2& \text{arithm. oper.}
  \end{align*}
  \begin{align*}
    \bexp :&:= true\ |\ false & \text{Boolean const.}\\
           &|\ \aexp_1 \opb \aexp_2 & \text{comparison}\\
           &|\ \bexp_1 \opl \bexp_2\ |\ \lnot \bexp_1 & \text{logic oper.}
  \end{align*}
\end{minipage}
\hspace{3mm}
\begin{minipage}{7cm}
\begin{align*}
  \com :&:= \wskip & \text{do nothing} \\
    &|\ \locval := \aexp  & \text{\!\!\!\!\!\!\!\!\!\!\!\!\!\!direct assignment}\\
    &|\,*\locval := \aexp  & \text{\!\!\!\!\!\!\!\!\!\!\!\!\!\!\!\!\!\!\!\!\!indirect assignment}\\
    &|\ \com_1;\com_2  & \text{sequence}\\
    &|\ \wassert{P}  & \text{assertion}\\
    &|\ \wif{\bexp}{\com_1}{\com_2}  & \text{condition}\\
    &|\ \wwhilei{\bexp}{P}{\com_1}  & \text{loop}\\
    &|\ \wcall{\loccom}  & \text{\!\!\!\!\!\!\!\!\!\!\!\!\!\!procedure call}
\end{align*}
\end{minipage}
\vspace{-1mm}
\caption{Syntax of arithmetic and Boolean expressions and commands in $\mathcal{L}$.}
%\vspace{-5mm}
\label{fig:exps-coms}
\end{figure}

Figure~\ref{fig:exps-coms} also presents the syntax of commands in $\mathcal{L}$.
Sequences, skip and conditions are standard.
An assignment can be done to a location directly or after a dereference.
Recall that a location $\locval_i$ contains as a value a natural number, 
say $v$, that can be seen in turn as the address of a location, namely $\locval_{v}$,
so the assignment $*\locval_i := \aexp$ writes the value of expression $a$ 
to the location $\locval_{v}$, 
while the address operation $\&\locval_i$ computes the address $i$ of $\locval_i$. 
An assertion command $\wassert{P}$ indicates that an assertion $P$ should
be valid at the point where the command occurs.
The loop command $\wwhilei{\bexp}{P}{\com_1}$ is always annotated with an invariant $P$.
As usual, this invariant should hold 
%\commentNK{Pas sur que cette phrase sur 'invariant must hold when we reach'
%soit utile ici sachant que l'execution l'ignore. Le sens de 'must' peut induire en erreur, a supprimer (ou a preciser)?. Je tente de remplacer 'must' par 'should' et de preciser ce role, meme si deja dit en sec.2.3.}
when we reach the command and be preserved by each loop step.
Command $\wcall{\loccom}$ is a procedure call. 
All annotations (assertions,  loop 
invariants and procedure contracts) will be ignored 
during the program execution 
and will be relevant only for program verification in Section~\ref{sec:verif-cond-gener}.
Procedures do not have
explicit parameters and return values (hence we use
the term \emph{procedure call} rather than \emph{function call}).
Instead, as in assembly code% 
%with a given calling convention 
~\cite{Irvine:2014:ALX:2655333}, parameters and
return value(s) are shared implicitly between the caller and the callee through memory locations:
the caller must put/read the right values at the right locations before/after
the call. 
%The called procedure is the one located at $\loccom$.
%For simplicity, we do not provide commands to update command bindings (or support function pointers): our set of routines will be fixed for each program.
Finally, to avoid ambiguity, we delimit sequences of commands with $\{\,\}$.

\begin{figure}[tb]
\centering
%\begin{subfigure}{2cm}
\begin{minipage}{2cm}
\begin{align*}
\crec\triangleq\enskip
\begin{array}{l@{}l}
&x_1:= x_4;\\
&x_2:= 0;\\
&\wcall{y_1}\\
\end{array}
\end{align*}
%\caption{command}\label{fig:com-example}
%\end{subfigure}
\end{minipage}
\hspace{2cm}
\begin{minipage}{6cm}
%\begin{subfigure}{6.5cm}
\begin{align*}
  \memcom = \left\{\enskip
  \loccom_1 \enskip \to \enskip
  \begin{array}{l@{}l}
    &\wif{x_1 > 0}{\\
    &\qquad x_2:= x_2 + x_3;\\
    &\qquad x_1:= x_1 -1;\\
    &\qquad \wcall{\loccom_1}\\
    &}{\\
    &\qquad \wskip\\
    &}\\
  \end{array}\,, \enskip \dots \enskip
  \right\}
\end{align*}
%\caption{procedure environment}
%\label{fig:proc-env-example}
%\end{subfigure}
\end{minipage}
\vspace{-1mm}
%\begin{subfigure}{12cm}
\begin{equation*}
\memann =
\left\{\,
\loccom_1 \,\to\,
\left(
\begin{array}{c}
  \lambda \mem. \mem(2) = \mem(3) \times (\mem(4)-\mem(1))
  \land 0 \leq \mem(1) \land \mem(1) \leq \mem(4),\\
  \lambda \mem. \mem(2) = \mem(3) \times \mem(4)
\end{array}
\right)\,,\enskip \dots \enskip
\right\}.
\end{equation*}
%\caption{contract environment}
%\label{fig:contract-env-example}
%\end{subfigure}
\vspace{-2mm}
\caption{Example of an \langname program $\crec$ with its environments.}
%\vspace{-5mm}
\label{fig:ex1}
\end{figure}

\begin{example}
\label{ex-mult}
%Figure~\ref{fig:proc-env-example} 
Figure~\ref{fig:ex1} shows an example of a command $\crec$ and a procedure environment $\memcom$
where procedure $\loccom_1$ points to a recursive command, called in $\crec$.
%the command.
% of Fig.~\ref{fig:com-example}.
With the semantics of Sec.~\ref{sec:oper-semant},
from any initial state,
the command will return a state in which $x_2 = x_3 \times x_4$. Procedure $\loccom_1$
returns a state where $x_2 = x_3 \times x_4$ if the
initial state satisfies $ x_2 = x_3 \times (x_4-x_1) \land 0 \leq x_1 \land x_1 \leq x_4$.
This can be expressed by the contract environment $\memann$ given (in 
$\lambda$-notation) in 
%Fig.~\ref{fig:contract-env-example}.
Fig.~\ref{fig:ex1}.
\hfil\qed
\end{example}

%\subsection{Evaluation of Arithmetic and Boolean expressions}
%\label{sec:eval-arithm-bool}

\begin{figure}[t]
\centering
\begin{minipage}{2.5cm}
\begin{align*}
  \EAexpf{\nat}{\mem} & \triangleq \nat
\end{align*}
\end{minipage}
\begin{minipage}{3cm}
	\begin{align*}
	\EAexpf{\locval_i}{\mem} & \triangleq \mem(i)
	\end{align*}
\end{minipage}
\begin{minipage}{3cm}
	\begin{align*}
	\EAexpf{*\locval_i}{\mem} & \triangleq \mem(\mem(i))
	\end{align*}
\end{minipage}
\begin{minipage}{3cm}
	\begin{align*}
	\EAexpf{\&\locval_i}{\mem} & \triangleq i
	\end{align*}
\end{minipage}
%\begin{minipage}{6cm}
%	\begin{align*}
%	\EAexpf{\nat}{\mem} & \triangleq \nat\\
%	\EAexpf{\locval_i}{\mem} & \triangleq \mem(i)\\
%	\EAexpf{*\locval_i}{\mem} & \triangleq \mem(\mem(i))\\
%	\EAexpf{\&\locval_i}{\mem} & \triangleq i\\
%	\EAexpf{\aexp_1 \opa \aexp_2}{\mem} & \triangleq
%	\EAexpf{\aexp_1}{\mem} \opa \EAexpf{\aexp_2}{\mem}
%	\end{align*}
%\end{minipage}
%\begin{minipage}{4cm}
%\begin{align*}
%  \EBexpf{true}{\mem} & \triangleq \True\\
%  \EBexpf{false}{\mem} & \triangleq \False\\
%  \EBexpf{\aexp_1 \opb \aexp_2}{\mem} & \triangleq
%      \EAexpf{\aexp_1}{\mem} \opa \EAexpf{\aexp_2}{\mem}\\
%  \EBexpf{\bexp_1 \opl \bexp_2}{\mem} & \triangleq
%      \EBexpf{\bexp_1}{\mem} \opl \EBexpf{\bexp_2}{\mem}\\
% \EBexpf{\neg\bexp}{\mem} & \triangleq  \neg\EBexpf{\bexp}{\mem}
%\end{align*}
%\end{minipage}
\vspace{-2mm}
\caption{Evaluation of expressions in $\mathcal{L}$ (selected rules).}
%\vspace{-5mm}
\label{fig:exp-eval}
\end{figure}

\begin{figure}[t]
%\scriptsize
%\begin{minipage}{3cm}
%  \begin{equation*}
%    % \tag{\textsc{skip}}
%    \inferrule*
%    {}
%    {\eval{\wskip}{\memvar}{\memcom}{\memvar}}
%  \end{equation*}
%\end{minipage}
\centering
%\begin{minipage}{150mm}
\centering
\hspace{-4mm}
\begin{minipage}{27mm}
%\begin{framed}
	\begin{equation*}
	%\tag{\textsc{assert}}
	\inferrule*
	% {P(\memvar)}
	{}
	{\eval{\wassert{P}}{\memvar}
		{\memcom}{\memvar}}
	\end{equation*}
%\end{framed}
\end{minipage} \qquad
\begin{minipage}{30mm}
%\begin{framed}
  \begin{equation*}
    % \tag{\textsc{assign}}
    \inferrule*
    {\EAexpf{\aexp}{\memvar} = \nat}
    {\eval{\locval_i := \aexp}{\memvar}{\memcom}{\memvar[i/n]}}
  \end{equation*}
%\end{framed} 
\end{minipage}

\begin{minipage}{35mm}
%\begin{framed}
  \begin{equation*}
    % \tag{\textsc{assign-mem}}
    \inferrule*
    {\EAexpf{\aexp}{\memvar} = \nat}
    {\eval{*\locval_i := \aexp}{\memvar}
      {\memcom}\memvar[\memvar(i)/n]}
  \end{equation*}
%\end{framed}
\end{minipage}\qquad
%\begin{minipage}{6cm}
%\begin{equation*}
%  %\tag{\textsc{if true}}
%  \inferrule*
%  {\EBexpf{\bexp}{\memvar} = \True\\
%    \eval{\com_1}{\memvar_1}{\memcom}\memvar_2}
%  {\eval{\wif{\bexp}{\com_1}{\com_2}}{\memvar_1}{\memcom}\memvar_2}
%\end{equation*}
%\end{minipage}
%
%\begin{minipage}{5cm}
%  \begin{equation*}
%    %\tag{\textsc{sequence}}
%    \inferrule*
%    {\eval{\com_1}{\memvar_1}{\memcom}\memvar_2\\
%      \eval{\com_2}{\memvar_2}{\memcom}\memvar_3}
%    {\eval{\com_1;\com_2}{\memvar_1}{\memcom}\memvar_3}
%  \end{equation*}
%\end{minipage}
%\begin{minipage}{6cm}
%\begin{equation*}
% % \tag{\textsc{if false}}
%  \inferrule*
%  {\EBexpf{\bexp}{\memvar} = \False\\
%    \eval{\com_2}{\memvar_1}{\memcom}\memvar_2}
%  {\eval{\wif{\bexp}{\com_1}{\com_2}}{\memvar_1}{\memcom}\memvar_2}
%\end{equation*}
%\end{minipage}
%\\[1mm]
%\vspace{-2mm}
%\begin{minipage}{12cm}
%\begin{equation*}
%  %\tag{\textsc{while true}}
%  \inferrule*
%  {\EBexpf{\bexp}{\memvar_1} = \True\\
%   \eval{\com_1}{\memvar_1}{\memcom}\memvar_2\\
%   \eval{\wwhilei{\bexp}{P}{\com}}{\memvar_2}{\memcom}\memvar_3}
%  {\eval{\wwhilei{\bexp}{P}{\com}}{\memvar_1}{\memcom}\memvar_3}
%\end{equation*}
%\end{minipage}
%\\[-1mm]
%\begin{minipage}{6cm}
%\begin{equation*}
%%  \tag{\textsc{while false}}
%  \inferrule*
%  {\EBexpf{\bexp}{\memvar} = \False}
%  {\eval{\wwhilei{\bexp}{P}{\com}}{\memvar}{\memcom}\memvar}
%\end{equation*}
%\end{minipage}
\begin{minipage}{33mm}
%\begin{framed}
\begin{equation*}
%  \tag{\textsc{call}}
  \inferrule*
  {\eval{\body{\loccom}{\memcom}}{\memvar_1}{\memcom}\memvar_2}
  {\eval{\wcall{\loccom}}{\memvar_1}{\memcom}\memvar_2}
\end{equation*}
%\end{framed}
\end{minipage}
%\end{minipage}
\vspace{-2mm}
\caption{Operational semantics of commands in $\mathcal{L}$ (selected rules).}
%\vspace{-2mm}
\label{fig:oper-sem}
\end{figure}

\vspace{-5mm}
\subsection{Operational Semantics}
\vspace{-2mm}
\label{sec:oper-semant}

%\commentLBi{Move Figure 4 and 5 to appendix}

Evaluation of arithmetic and Boolean expressions in $\mathcal L$ is defined by functions
$\EAexp$ and $\EBexp$.
Selected evaluation rules for arithmetic expressions are shown in Fig.~\ref{fig:exp-eval}.
%As mentioned above, the subtraction is lower-bounded by 0.
Operations $*\locval_i$ and $\&\locval_i$ have a
semantics similar to the C language, {i.e.} dereferencing and address-of.
Semantics of Boolean expressions is standard~\cite{DBLP:books/daglib/0070910}.

Based on these evaluation functions, we can define the operational semantics
of commands %of \langname 
in a given procedure environment $\memcom$.
%as shown in Fig.~\ref{fig:oper-sem}.
Selected evaluation rules\footnote{For convenience of the reader,
full versions of Fig.~\ref{fig:exp-eval},~\ref{fig:oper-sem} are given in Appendix~\ref{app-sem}.}
are shown in Fig.~\ref{fig:oper-sem}.
As said above, both assertions and loop invariants can be seen
as program annotations that
do not influence the execution of the program itself.
Hence, command $\wassert{P}$ is equivalent to a skip. Likewise,
loop invariant $P$ has no influence
on the semantics of $\wwhilei{\bexp}{P}{\com}$.

%In the following 
We write $\Vdash\eval{\com}{\memvar}{\memcom}{\memvar'}$
to denote that $\eval{\com}{\memvar}{\memcom}{\memvar'}$
can be derived from the rules of Fig.~\ref{fig:oper-sem}.
Our \coq formalization,
inspired by \cite{SF},
provides a deep embedding of \langname,
%\commentNK{Maybe a more explicit sentence on what is
%modeled in which file(s) would be better?}
with an associated parser, in  files
\texttt{Aexp.v}, \texttt{Bexp.v} and \texttt{Com.v}.

%% file: funct-corr.tex
We define functional correctness in a similar way to the original
\emph{Hoare triple} definition~\cite{DBLP:journals/cacm/Hoare69},
except that we also need a procedure environment $\memcom$, leading
to a quadruple denoted $\hoared{P}{\com}{Q}{\memcom}$. We will however
still refer by the term ``Hoare triple'' to the corresponding
program property, formally defined as follows.
\begin{definition}[Hoare triple]
\label{def:funct-corr-1}
Let $\com$ be a command, $\memcom$ a procedure environment,
and $P$ and $Q$ two assertions. We define a Hoare triple $\hoared{P}{\com}{Q}{\memcom}$
as follows:
\[
\hoared{P}{\com}{Q}{\memcom}\,
\triangleq\,
  \forall \memvar, \memvar' \in \Memvar.\  \assert{P}{\memvar} \land
  (\Vdash\eval{\com}{\memvar}{\memcom}{\memvar'}) \Rightarrow \assert{Q}{\memvar'}.
\]

\end{definition}

Informally, our definition states
that, for a given $\memcom$, if a state $\memvar$ satisfies $P$ and
the execution of $\com$ on $\memvar$ terminates in a state $\memvar'$,
then $\memvar'$ satisfies $Q$.

Next, we introduce notation $\hoarectx{\memcom}{\memann}$
to denote the fact that, 
for the given $\memann$ and  $\memcom$,
every procedure %call 
satisfies its contract.
%$\memann$ for a given body $\memcom$. 

\begin{definition}[Contract Validity]
Let $\memcom$ be a procedure environment and $\memann$ a contract environment.
We define contract validity $\hoarectx{\memcom}{\memann}$ as follows:
\[\hoarectx{\memcom}{\memann} \,\triangleq\,
  \forall \loccom \in \Loccom.\
  \hoared{\pre{\loccom}{\memann}}{\wcall{\loccom}}{\post{\loccom}{\memann}}{\memcom}).
\]
\end{definition}

The notion of contract validity is at the heart of
modular verification, since it allows assuming that the contracts
of the callees are satisfied during the verification of a Hoare triple.
More precisely, to state the validity of procedure contracts
without assuming anything about their bodies in our formalization,
we will consider an arbitrary choice of implementations $\memcom'$ of procedures
that satisfy the contracts, like in assumption (\ref{hyp:lemma1}) in Lemma \ref{hoare:recursionproc}.
%(as the operational semantics based definition would force us to do),
This technical lemma, taken from \cite[Equation (4.6)]{DBLP:series/txcs/AptBO09},
gives an alternative criterion for validity of procedure contracts: 
%if for each procedure $\loccom$
%we can deduce that its body respects its contract
%from the assumption that the calls of all procedures respect their
%contracts (without assuming anything about their bodies, i.e. 
%\commentNK{I tried to give a closer intuition... Is it OK?}
%for an arbitrary choice of implementations $\memcom'$ of procedures),
%then all procedure calls respect their contracts.
if, under the assumption that the contracts in $\memann$ hold, we can prove
for each procedure $y$ that its body satisfies its contract, then the contracts are
valid.

\vspace{-2mm}
\begin{lemma}[Adequacy of contracts]
  \label{hoare:recursionproc}
  Given  a procedure environment $\memcom$
  and a contract environment $\memann$ such that 
  \begin{equation}\label{hyp:lemma1}
    \forall \memcom' \in \Memcom.\ 
    \hoaredctx{\pre{\loccom}{\memann}}
    {\body{\loccom}{\memcom}}
    {\post{\loccom}{\memann}}{\forall \loccom \in \Loccom, \memcom'}{\hoarectx{\memcom'}{\memann}},
  \end{equation}
  we have \enskip
  $%\begin{equation*}
    \hoarectx{\memcom}{\memann}.
  $%\end{equation*}
\end{lemma}

\begin{proof}%[Sketch]
%NK:	
Any given terminating execution traverses a finite number of procedure calls
(over all procedures)
that can be replaced by inlining the bodies a sufficient number of times.
We first formalize a theory of $k$-inliners (that inline procedure bodies 
a finite number of times $k\ge 0$ and replace deeper calls 
by nonterminating loops) and prove their properties.
Relying on this elegant theory, the proof of the lemma proceeds by induction on 
the number of procedure inlinings.
%By assuming the Hoare triple property 
%for executions with at most $n$ inlinings of procedures, 
%we prove the result also for a state with $n+1$ inlinings by induction.
%LB,VP:  
%  By induction over the number of procedure inlining in the procedure environment
%  $\memcom$.
  \qed
\end{proof}

From that, we can establish the main result of this section. Theorem \ref{hoare:recursion}, 
taken from \cite[Th. 4.2]{DBLP:series/txcs/AptBO09} states that
$\hoared{P}{\com}{Q}{\memcom}$ holds if
assumption (\ref{hyp:lemma1}) holds and if 
the validity of contracts of $\memann$ for $\memcom$ 
implies the Hoare triple itself.
%assuming the contracts $\memann$, the body of each procedure $y$ satisfies the
%corresponding contract,
%and, always under the assumption that $\memann$ holds,
%the triple itself $\hoared{P}{\com}{Q}{\memcom}$ is valid.
This theorem is the basis for modular verification of Hoare Triples,
as done for instance in Hoare Logic \cite{DBLP:journals/cacm/Hoare69,DBLP:books/daglib/0070910}
or verification condition generation. 
%Note that working with these contracts
%makes the exact implementation of called procedures irrelevant: 
%in the first hypothesis
%of Theorem~\ref{hoare:recursion}, as well as in Lemma~\ref{hoare:recursionproc},
%we could replace the bodies of the callees
%of $y$ by any command that satisfies the corresponding contract in $\phi$.

\begin{theorem}[Recursion]
  \label{hoare:recursion}
  Given a procedure environment $\memcom$
  and a contract environment $\memann$ such that 
  \begin{gather*}
    \forall \memcom' \in \Memcom.\ 
    \hoaredctx{\pre{\loccom}{\memann}}
    {\body{\loccom}{\memcom}}
    {\post{\loccom}{\memann}}{\forall \loccom \in \Loccom, \memcom'}{\hoarectx{\memcom'}{\memann}}, \ \mbox{and}\\    
    \hoaredctx{P}{\com}{Q}{\memcom}{\hoarectx{\memcom}{\memann}},
  \end{gather*}
  we have \enskip
  $%\begin{equation*}
    \hoared{P}{\com}{Q}{\memcom}.
  $%\end{equation*}
\end{theorem}

\begin{proof}
  By Lemma \ref{hoare:recursionproc}.\qed
\end{proof}

%This theorem is classical~\cite{DBLP:series/txcs/AptBO09}
%and 
We refer the reader to the \coq development,
more precisely the results 
\verb|recursive_proc| and
\verb|recursive_hoare_triple| in file
\verb|Hoare_Triple.v| for complete proofs of 
Lemma~\ref{hoare:recursionproc} and
Theorem~\ref{hoare:recursion} for \langname.
To the best of our knowledge, this is the first mechanized proof of 
these classical results. % in \coq.

An interesting corollary can
be deduced from Theorem~\ref{hoare:recursion}.

\vspace{-2mm}
\begin{cor}[Procedure Recursion]
  \label{hoare:procrecursion}
  Given a procedure environment $\memcom$
  and a contract environment $\memann$ such that 
  \begin{equation*}
    \forall \memcom' \in \Memcom.\ 
    \hoaredctx{\pre{\loccom}{\memann}}
    {\body{\loccom}{\memcom}}
    {\post{\loccom}{\memann}}{\forall \loccom \in \Loccom, \memcom'}{\hoarectx{\memcom'}{\memann}},
  \end{equation*}
  we have \enskip
  $%\begin{equation*}
     \forall \loccom \in \Loccom. \ \hoared{\pre{\loccom}{\memann}}
    {\body{\loccom}{\memcom}}
    {\post{\loccom}{\memann}}{\memcom}.
  $%\end{equation*}
\end{cor}
%\commentNK{I proposed to add a quantifier over $y$ into the conclusion.}

%% file: rel-prop.tex
Relational properties can be seen as an extension of
Hoare triples.
But, instead of linking one program with two properties, the pre- and
postconditions, relational
properties link $n$ programs to two properties, called relational assertions.
We define a \emph{relational assertion} 
as a predicate taking a sequence of memory states
and returning a first-order logic formula.
We use metavariables $\rela{P}, \rela{Q}, ...$
to range over the set of relational assertions, denoted $\rAssertion$.
As a simple example of a relational assertion, we might say
that two states bind location $x_3$ to the same value. This would be
stated as follows:
$\lambda (\memvar_1,\memvar_2). \memvar_1(3) = \memvar_2(3)$.

A \emph{relational property} is a property about $n$
programs $c_1, ..., c_n$, stating that if each program $c_i$
starts in a state $\memvar_i$ and ends in a state $\memvar'_i$ such
that $\rassert{P}{\memvar_1, ..., \memvar_n}$ holds,
then $\rassert{Q}{\memvar'_1, ..., \memvar'_n}$ holds,
where $\rela{P}$ and $\rela{Q}$ are relational assertions over $n$ memory states.

%\begin{equation*}
%  \sequence{u}{\nat} \triangleq
%  \begin{cases}
%    (x_1, ..., x_\nat) & \textrm{if}\ 0 < \nat\\
%    [\ ] & \textrm{else}
%  \end{cases}
%\end{equation*}

We formally define relational
correctness similarly to functional correctness (cf. Def. \ref{def:funct-corr-1}),
except that we now use sequences of memory states and commands of {equal length}.
We denote by  $\sequence{u}{\nat}$ 
a sequence of elements $(u_k)^\nat_{k=1}=(u_1,\dots,u_n),$  %of kind $X$, 
where
$k$ ranges from $1$ to $\nat$.
If $\nat \le 0$, $\sequence{u}{\nat}$ is the empty sequence denoted $[\ ]$.

\begin{definition}[Relational Correctness] \label{def:rela}
Let $\memcom$ be a procedure environment,
$\sequence{\com}{\nat}$  a sequence of $\nat$ commands ($\nat \in \Natset^*$),
and $\rela{P}$ and $\rela{Q}$ two relational assertions over $\nat$ states. 
The relational correctness of $\sequence\com{n}$ with respect to $\rela{P}$ and
$\rela{Q}$, denoted
$\rhoared{P}{\sequence\com{n}}{Q}{\memcom}$, is defined as follows:
\begin{gather*}
\rhoared{P}{\sequence\com{n}}{Q}{\memcom} \triangleq\\
\forall \sequence{\memvar}{\nat}, \sequence{\memvar'}{\nat}.\
    \rassert{P}{\sequence{\memvar}{\nat}} 
%    \Rightarrow
     \land
    (\evalrr{\com}{\memvar}{\memvar'}{\psi}{\nat})
    \Rightarrow
    \rassert{Q}{\sequence{\memvar'}{\nat}}.
\end{gather*}
\end{definition}

This notation generalizes the one proposed by Benton 
\cite{DBLP:conf/popl/Benton04}
for relational properties 
linking two 
%(in their work, typically, similar) 
commands:
%\begin{equation*}
$  \rhoared{P}{c_1\sim c_2}{Q}{\memcom}. $
%\end{equation*}
As Benton's work mostly focused on comparing equivalent programs, using symbol $\sim$
%to denote a relation of similarity between two programs 
was quite natural.
%For example, let us consider the following relational assertions.
In particular, Benton's work would not be practical for verification of 
relational properties with several calls such as transitivity mentioned in Sec.~\ref{sec:intro}.

\begin{figure}[tb]
\centering
\begin{minipage}{12cm}
	\normalsize
\begin{equation*}\label{notation:ben}
\memcom :
\,\,\{\rela{P}\}\,\,
\cswI
%\ben{\cswI}{1}
\,\,\sim\,\,
\cswII
%\ben{\cswII}{2}
\,\,\{\rela{Q}\},
\end{equation*}
\vspace{-6mm}
\begin{gather*}
  \rela{P}  \triangleq\
  \lambda\memvar_1\memvar_2.\ 
  \memvar_1(\memvar_1(1)) = \memvar_2(\memvar_2(1)) \land
  \memvar_1(\memvar_1(2)) = \memvar_2(\memvar_2(2)) \land \\ 
  \memvar_1(1) \neq \memvar_1(2) \land  
  \memvar_2(1) \neq \memvar_2(2) \land   
  \memvar_1(1) > 3 \land
  \memvar_1(2) > 3 \land
  \memvar_2(1) > 2 \land
  \memvar_2(2) > 2,
\end{gather*}
\begin{equation*}
\rela{Q}\triangleq
\lambda\memvar'_1\memvar'_2.\ 
  \memvar'_1(\memvar'_1(1)) = \memvar'_2(\memvar'_2(1)) \land
  \memvar'_1(\memvar'_1(2)) = \memvar'_2(\memvar'_2(2)).
\end{equation*}
%NK:Was
%\begin{equation*}\label{notation:ben}
%  \memcom :
%  \,\,\left\{\rela{P}\right\}\,\,
%  \ben{
%    \begin{array}{l@{}l}
%    \locval_3 &:= *\locval_1;\\
%    *\locval_1 &:= *\locval_2;\\
%    *\locval_2 &:= \locval_3;
%    \end{array}
%  }{1}
%  \,\,\sim\,\,
%  \ben{
%  \begin{array}{l@{}l}
%    *\locval_1 &:= *\locval_1+ *\locval_2;\\
%    *\locval_2 &:= *\locval_1 - *\locval_2;\\
%    *\locval_1 &:= *\locval_1 - *\locval_2
%  \end{array}
%  }{2}
%   \,\,\left\{\rela{Q}\right\}.
%\end{equation*}
\end{minipage}
\vspace{-1mm}
\caption{A relational property for  $\mathcal{L}$ programs $\cswI$ and $\cswII$ 
of Fig.~\ref{fig:ex-rel-prop-swaps-self}.
%swapping  $*\locval_1$ and $*\locval_2$. 
%\commentNK{can be $\memvar_2(i) > 2$ instead of $3$ in $\rela{P}$??}
}
\vspace{-5mm}
\label{fig:ex-rel-prop-swaps}
\end{figure}

\begin{example}[Relational property]
\label{ex:rel-prop}
%NK: "Figure" MUST be non abbreviated when FIRST word of a sentence
Figure~\ref{fig:ex-rel-prop-swaps} formalizes the relational property $\Rsw$ for  $\mathcal{L}$ programs $\cswI$ and $\cswII$ 
discussed in Ex.~\ref{ex:motiv-ex}.
Recall that $\Rsw$ 
(written in Fig.~\ref{fig:ex-rel-prop-swaps} 
in Benton's notation) states that both programs
%(with tag $\ben{}{1}$ on 
%the left and tag $\ben{}{2}$ on the right) 
executed
from two states named $\memvar_1$ and $\memvar_2$ having the same values in
$*\locval_1$ and $*\locval_2$ will end up in two states $\memvar'_1$ and $\memvar'_2$ also
having the same values in these locations. Notice that the initial
state of each program needs separation hypotheses (cf. the second line of the definition of $\rela{P}$). 
Namely, $\locval_1$ and $\locval_2$ must point to different locations 
and must not point
to $\locval_1$, $\locval_2$ or, for $\cswI$, to $\locval_3$ for the property to hold.
This relational property is formalized 
%and proven valid 
in
the \coq development in file {\tt Examples.v}.
\qed
\end{example}

%% file: vcg.tex
A standard way~\cite{Floyd1967} for verifying that a Hoare triple holds is to use
a verification condition generator (VCGen).
In this section, we formalize a VCGen for Hoare triples
%introduced in Section~\ref{sec:funct-corr} 
and show that it is correct,
in the sense that if all verification conditions that it generates
are valid, then the Hoare triple is valid according to Def.~\ref{def:funct-corr-1}.

%\commentLB{Compute the strongest post-condition}

\vspace{-3mm}
\subsection{Verification Condition Generator}
\vspace{-2mm}
\label{vcgdef}

We have chosen to split the VCGen in three steps,
as it is commonly done \cite{Frama-C}:

\begin{itemize}
\item function $\Tcn$ generates the main verification condition,
expressing that the postcondition holds in the final state,
assuming auxiliary annotations hold;
\item function $\Tan$ generates auxiliary verification conditions stemming from
  assertions, loop invariants, and preconditions of called procedures;
\item finally, function $\Tfn$ generates verification conditions for the
auxiliary procedures that are called by the main program, to ensure that
their bodies respect their contracts.
\end{itemize}

\begin{figure}[tb]
\centering
\begin{minipage}{12cm}
%\normalsize  
\begin{align*}
  \Tcnf{\wskip}{\mem}{\memann}{f} & \triangleq  \forall \mem'.\, \mem' = \mem \Rightarrow f(\mem')\\
  \Tcnf{\locval_i := \aexp}{\mem}{\memann}{f} & \triangleq
  \forall \mem'.\, \mem' = \setmem{\mem}{i}{\EAexpf{\aexp}{\mem}}\Rightarrow f(\mem')\\
  \Tcnf{*\locval_i := \aexp}{\mem}{\memann}{f} & \triangleq
  \forall \mem'.\, \mem' = \setmem{\mem}{\mem(i)}{\EAexpf{\aexp}{\mem}} \Rightarrow f(\mem')\\
  \Tcnf{\wassert{P}}{\mem}{\memann}{f} & \triangleq
  \forall \mem'.\, \mem' = \mem \land P(\mem) \Rightarrow f(\mem')\\
  \Tcnf{\com_0;\com_1}{\mem}{\memann}{f} & \triangleq
                        \Tcnf{\com_0}{\mem}{\memann}{\lambda \mem'.\,
                          \Tcnf{\com_1}{\mem'}{\memann}{f}}\\
  \begin{split}
  \Tcnf{\wif{\bexp}{\com_0}{\com_1}}{\mem}{\memann}{f} & \triangleq
  (\EBexpf{\bexp}{\mem} \Rightarrow \Tcnf{\com_0}{\mem}{\memann}{f}) \land \\
                          & \qquad (\neg \EBexpf{\bexp}{\mem} \Rightarrow \Tcnf{\com_1}{\mem}{\memann}{f})
  \end{split}\\
  \Tcnf{\wcall{\loccom}} {\mem}{\memann}{f} & \triangleq
          \pre{\loccom}{\memann}(\mem) \Rightarrow (\forall \mem'.\, \post{\loccom}{\memann}(\mem') \Rightarrow f(\mem'))\\
  \Tcnf{\wwhilei{\bexp}{inv}{\com}}{\mem}{\memann}{f} & \triangleq
        inv(\mem) \Rightarrow \\ & \qquad
        (\forall \mem'.\, inv(\mem') \land \neg(\EBexpf{\bexp}{\mem'}) \Rightarrow f(\mem'))
\end{align*}
\end{minipage}
\vspace{-1mm}
\caption{Definition of function $\Tcn$ generating the main verification condition. 
%using the notation $\Tcnf{command \; c}{memory \; state \; \mem }{contract \; environment \; \memann}{assertion \; f}$.
}
%\vspace{-5mm}
\label{fig:def-Tc}
\end{figure}

\begin{definition}[Function $\Tcn$ generating the main verification condition]
\label{def:tc-fun}
Given a command $\com$, a memory state $\sigma$ representing the
state before the command, a contract environment $\memann$, and an assertion $f$,
function $\Tcn$ returns a formula defined by case analysis on $\com$ as shown in
Fig.~\ref{fig:def-Tc}.
\end{definition}

Assertion $f$ represents the postcondition we want to verify after the command executed from state $\mem$.
For each command, except sequence and branch, a fresh memory state $\mem'$ is
introduced and related to the current memory state $\mem$.
The new memory state is given as parameter to $f$.
For $\wskip$, which does nothing, both
states are identical. For assignments, $\mem'$ is simply the update of $\mem$. An assertion
introduces a hypothesis over $\mem$ but leaves it unchanged.
For a sequence,
we simply compose the conditions, that is, we check that the final state of $c_0$
is such that $f$ will be verified after executing $c_1$. For a conditional, we  check
that if the condition evaluates to true, the \emph{then} branch will ensure the postcondition,
and that otherwise the \emph{else} branch will
ensure the postcondition. The rule for calls simply assumes that $\mem'$ verifies $\post{y}{\phi}$. Finally,
$\Tcn$ assumes that, after a loop, $\sigma'$ is a state where the loop condition is
false and the loop invariant holds.
As for an assertion, the callee's precondition and the loop invariant
are just assumed to be true; function
$\Tan$, defined below, generates the corresponding proof obligations.
% before the command 
%(their validity will be treated by function $\Tan$).

\begin{example}
\label{ex:Tc}	
%For example, with 
For $\com \triangleq \wskip;\locval_1:= 2$, and
$f \triangleq \lambda \mem.\ \mem(1) = 2$, we have:
\begin{equation*}
  \Tcnf{\com}{\mem}{\memann}{f} \equiv
  \forall \mem'_1. \mem = \mem'_1 \Rightarrow
  (\forall \mem'_2. \mem'_2 = \setmem{\mem'_1}{1}{2} \Rightarrow
  \mem'_2(1) = 2). \qquad \qed
\end{equation*}

\end{example}

%Note that in $\Tcn$, we assume that assertions, pre-conditions and post-conditions
%of called procedures, as well as loop invariants, hold. $\Tan$, defined below, generates
%the corresponding proof obligations.

\begin{figure}[tb]
\centering
\begin{minipage}{12cm}
%\normalsize  
\begin{align*}
  \Tanf{\wskip}{\mem}{\memann} & \triangleq True\\
  \Tanf{\locval_i := \aexp}{\mem}{\memann} & \triangleq True\\
  \Tanf{*\locval_i := \aexp}{\mem}{\memann} & \triangleq True\\
  \Tanf{\wassert{P}}{\mem}{\memann} & \triangleq P(\mem)\\
  \begin{split}
    \Tanf{\com_0;\com_1}{\mem}{\memann} & \triangleq \Tanf{\com_0}{\mem}{\memann} \land \\
    & \Tcnf{\com_0}{\mem}{\memann}{\lambda \mem'. (\Tanf{\com_1}{\mem'}{\memann})}
  \end{split}\\
  \begin{split}
    \Tanf{\wif{\bexp}{\com_0}{\com_1}}{\mem}{\memann} & \triangleq
    \EBexpf{\bexp}{\mem} \Rightarrow \Tanf{\com_0}{\mem}{\memann} \land \\
    & \neg(\EBexpf{\bexp}{\mem}) \Rightarrow \Tanf{\com_1}{\mem}{\memann}
  \end{split}\\
  \Tanf{\wcall{\loccom}}{\mem}{\memann} & \triangleq \pre{\loccom}{\memann}(\mem)\\
  \begin{split}
  \Tanf{\wwhilei{\bexp}{inv}{\com}}{\mem}{\memann} & \triangleq inv(\mem) \land \\
  & (\forall \mem', inv(\mem')  \land
    \EBexpf{\bexp}{\mem'} \Rightarrow \Tanf{c}{\mem'}{\memann}) \land \\
  & (\forall \mem' , inv(\mem') \land  \EBexpf{\bexp}{\mem'}\Rightarrow
                     \Tcnf{\com}{\mem'}{\memann}{inv})
 \end{split}
\end{align*}
\end{minipage}
\vspace{-1mm}
\caption{Definition of function $\Tan$ generating auxiliary verification conditions.}
%\vspace{-5mm}
\label{fig:def-Ta}
\end{figure}

\vspace{-3mm}
\begin{definition}[Function $\Tan$ generating the auxiliary verification conditions]
\label{def:ta-fun}
Given a command $\com$, a memory state $\mem$ representing the
state before the command, and a contract environment $\memann$,
function $\Tan$ returns a formula defined by case analysis on $\com$
as shown in Fig.~\ref{fig:def-Ta}.
\end{definition}

Basically, $\Tan$ collects all assertions, preconditions of called procedures, as
well as invariant establishment and preservation, and lifts the corresponding
formulas to constraints on the initial state $\mem$ through the use of $\Tcn$.

\smallskip
Finally, we define the function for generating the conditions
for verifying that the body of each procedure defined in $\memcom$
respects its contract defined in $\memann$.
\begin{definition}[Function $\Tfn$ generating the procedure verification condition]
\label{def:tf-fun}
Given two environments $\memcom$ and $\memann$, $\Tfn$ returns
the following formula:
\begin{equation*}
  \begin{split}
  \Tfnf{\memann}{\memcom} \triangleq
  \forall \loccom, \mem.\,\,\pre{\loccom}{\memann}(\mem) \Rightarrow\,\, &
  \Tanf{\body{\loccom}{\memcom}}{\mem}{\memann} \,\land \\
  & \Tcnf{\body{\loccom}{\memcom}}{\mem}{\memann}{\post{\loccom}{\memann}}.
  \end{split}
\end{equation*}
\end{definition}

The VCGen is defined in file {\tt Vcg.v} of the \coq development.
Interested readers will also find a proof 
%an extra proof  
%We also prove 
(in file {\tt Vcg_Opt.v}) of a VCGen 
optimization (not detailed here), % and simplicity.
which 
%One of the optimizations~\cite{DBLP:conf/popl/FlanaganS01}
prevents the size of the generated formulas
from becoming exponential in the number of
conditions in the program~\cite{DBLP:conf/popl/FlanaganS01}, which is a classical problem
for ``naive'' VCGens.

\vspace{-3mm}
\subsection{Hoare Triple Verification}
\vspace{-2mm}

We can now state 
%give the (classic) 
the theorems establishing correctness
of the VCGen. Their proof can be found in file
{\tt Correct.v} of the \coq development.

First, Lemma~\ref{the:vcghoarectx} shows that,
under the assumption of the procedure contracts, a Hoare
triple is valid if for all memory states satisfying the precondition,
the main verification condition and the auxiliary verification
conditions hold.

\begin{lemma}
  \label{the:vcghoarectx}
  Assume the following two properties hold:
  \begin{gather*}
    \forall \mem \in \Sigma, \assert{P}{\mem} \Rightarrow \Tanf{\com}{\mem}{\memann},\\
    \forall \mem \in \Sigma, \assert{P}{\mem} \Rightarrow \Tcnf{\com}{\mem}{\memann}{Q}.
  \end{gather*}
  Then we have \enskip 
  $%\begin{equation*}
    \hoaredctx{P}{\com}{Q}{\memcom}{\hoarectx{\memcom}{\memann}}.
  $%^\end{equation*}
\end{lemma}

\begin{proof}
  By structural induction over $\com$.
  \qed
\end{proof}

Next, we prove in Lemma~\ref{the:vcghoareproc} that if $\Tfnf{\memann}{\memcom}$ holds, then 
for an arbitrary choice of implementations $\memcom'$ of procedures
respecting the procedure contracts, the body  of each procedure 
$y$ respects its contract.

\begin{lemma}
  \label{the:vcghoareproc}
  Assume that the formula
%  \begin{equation*}
   $\Tfnf{\memann}{\memcom}$ is satisfied.
%  \end{equation*}
  Then we have
  \begin{equation*}
    \forall \memcom' \in \Memcom.\ 
    \hoaredctx{\pre{\loccom}{\memann}}
    {\body{\loccom}{\memcom}}
    {\post{\loccom}{\memann}}{\forall \loccom \in \Loccom,\memcom'}{\hoarectx{\memcom'}{\memann}}.
  \end{equation*}
\end{lemma}

\begin{proof}
  By Lemma~\ref{the:vcghoarectx}.
  \qed
\end{proof}

Finally, we can establish the main theorem of this section, stating that
the VCGen is correct with respect to our definition of Hoare triples.

\begin{theorem}[Soundness of VCGen]
  \label{hoare:proof}
  Assume that %the formula
%  \begin{equation*}
we have
$\Tfnf{\memann}{\memcom}$ and 
%is satisfied and that we have
  \begin{gather*}
    \forall \mem \in \Sigma, \assert{P}{\mem} \Rightarrow \Tanf{\com}{\mem}{\memann}, \\
    \forall \mem \in \Sigma, \assert{P}{\mem} \Rightarrow \Tcnf{\com}{\mem}{\memann}{Q}.
  \end{gather*}
  Then we have
%  \begin{equation*}
    $\hoared{P}{\com}{Q}{\memcom}$.
%  \end{equation*}
\end{theorem}

\begin{proof}
  By Theorem~\ref{hoare:recursion} and Lemmas \ref{the:vcghoarectx} and \ref{the:vcghoareproc}.
\qed
\end{proof}

\begin{example}
Consider again the command $\crec$, procedure environment $\memcom$, and contract
environment $\memann$ of Ex.~\ref{ex-mult}
(presented in Fig.~\ref{fig:ex1}).
We can apply Theorem~\ref{hoare:proof} to
prove its functional correctness expressed by the following Hoare triple:
\begin{equation*}
  \memcom: \{\lambda \mem. True\}\enskip
  \crec\enskip
   \{\lambda \mem. \mem(2) = \mem(4) \times \mem(3)\}
\end{equation*}
(see command \verb'com_rec' in file {\tt Examples.v}).
%\commentNK{Should be com\_rec in Coq?}
% The proof relies on Theorem~\ref{hoare:proof}.
\qed
\end{example}

%% file: verifrelaprop.tex
In this section, we propose a verification method for relational
properties (defined in Section~\ref{sec:rel-prop})
using the VCGen defined in Section~\ref{sec:verif-cond-gener}
(or, more
generally, any VCGen respecting Theorem~\ref{hoare:proof}).
%\subsection{Verification Conditions for Relational Properties}
First, we define the notation $\Tr$ for the recursive call of function
$\Tcn$ on a sequence of commands and memory states:

\begin{definition}[Function $\Tr$]
  \label{def:Trf}
%  $\Tr$ takes as arguments two sequences, 
  Given  a sequence of commands $\sequence\com\nat$ and 
  a sequence of memory states $\sequence\mem\nat$,
%  (both of size $\nat$), 
  a contract environment
  $\memann$ and a predicate $\rela{Q}$ over $\nat$ states,
  function $\Tr$ is defined by induction on $\nat$ 
  %over  the number of elements 
  as follows.
  \begin{itemize}
  \item Basis: $\nat = 0$.
    \begin{equation*}
      \Trf{[\ ]}{[\ ]}{\memann}{\rela{Q}} \triangleq \rassert{Q}{[\ ]}.
    \end{equation*}
  \item Inductive: $\nat \in \Natset^*$.
    \begin{gather*}
      \Trf{\sequence{\com}{\nat}}
      {\sequence{\mem}{\nat}}{\memann}{\rela{Q}} \triangleq\\
      \Tcnf{\com_\nat}{\mem_\nat}{\memann}
      {\,\,\lambda \mem_\nat'.
        \Trf{\sequence{\com}{\nat-1}}
        {\sequence{\mem}{\nat-1}}{\memann}
        {\,\,\lambda \sequence{\mem'}{\nat-1}. \rassert{Q}{\sequence{\mem'}{\nat}}}}.
    \end{gather*}
  \end{itemize}
\end{definition}

Intuitively, for $\nat=2$, $\Tr$ gives the weakest relational condition that $\mem_1$ and $\mem_2$ must fulfill in
order for $\rela{Q}$ to hold after executing $\com_1$ from $\mem_1$ and
$\com_2$ from $\mem_2$:
%\commentNK{Donner la formule serait utile?}
$
\Trf{(\com_1, \com_2)}{(\mem_1, \mem_2)}{\memann}{\rela{Q}} \equiv
\Tcnf{\com_2}{\mem_2}{\memann}
{\,\lambda \mem'_2. \Tcnf{\com_1}{\mem_1}{\memann}
	{\,\lambda \mem'_1. \rassert{Q}{\mem_1',\mem'_2}}}.
$

\begin{remark}
  \label{rmk:Trf}
  Assume we have $\nat>0,$ a command $\com_n$,
  a sequence of commands
  $\sequence{\com}{\nat-1}$, and a sequence of memory states
  $\sequence{\mem}{\nat-1}$. 
  From Def.~\ref{def:funct-corr-1}, it follows that
  \begin{gather*}
    \forall \mem_\nat, \mem'_\nat.\
    \rassert{P}{\sequence{\mem}{\nat}}
%    \Rightarrow
    \land
    (\Vdash\eval{\com_\nat}{\mem_\nat}{\memcom}{\mem'_\nat})
    \Rightarrow\\
    \Trf{\sequence{\com}{\nat-1}}{\sequence{\mem}{\nat-1}}{\memann}
    {\lambda \sequence{\mem'}{\nat-1}.\rassert{Q}{\sequence{\mem'}{\nat}}}
  \end{gather*}
  is equivalent to
  \begin{equation*}
    \hoared{\lambda \mem_n. \rassert{P}{\sequence{\mem}{\nat}}}
    {\com_\nat}
    { \lambda \mem'_\nat.
      \Trf{\sequence{\com}{\nat-1}}{\sequence{\mem}{\nat-1}}{\memann}
      {\lambda \sequence{\mem'}{\nat-1}.\rassert{Q}{\sequence{\mem'}{\nat}}}}
    {\memcom}.
  \end{equation*}
\end{remark}

\begin{example}[Relational verification condition]
\label{ex:rel-prop-vc}
In order to make things more concrete, we can go back 
to the relational property $\Rsw$
%equivalence 
between
two implementations $\cswI$ and $\cswII$ of \texttt{swap} defined in Ex.~\ref{ex:motiv-ex} and 
examine what would be the main verification condition generated by $\Tr$.	
%Let us denote:
%\[c_1\triangleq
%  \begin{array}{l@{}l}
%    \locval_3 &:= *\locval_1;\\
%    *\locval_1 &:= *\locval_2;\\
%    *\locval_2 &:= \locval_3;
%    \end{array}
%\qquad\mathrm{and}\qquad
%c_2\triangleq
%  \begin{array}{l@{}l}
%    *\locval_1 &:= *\locval_1+ *\locval_2;\\
%    *\locval_2 &:= *\locval_1 - *\locval_2;\\
%    *\locval_1 &:= *\locval_1 - *\locval_2
%  \end{array}
%\]
%and 
Let $\rela{P}$ and $\rela{Q}$ be defined as in Ex.~\ref{ex:rel-prop}.
In this particular case, we have $n=2$, and $\memann$ is empty  (since
we do not have any function call), thus Def.~\ref{def:Trf} becomes:
\[
  \Trf{(\cswI, \cswII)}{(\mem_1, \mem_2)}{\emptyset}{\rela{Q}} \!=\!
  \Tcnf{\cswII}{\mem_2}{\emptyset}
  {\lambda \mem'_2. \Tcnf{\cswI}{\mem_1}{\emptyset}
    {\lambda \mem'_1. \rassert{Q}{\mem_1',\mem'_2}}}.
\]
We thus start by applying $\Tcn$ over $\cswI$, to obtain, using the rules of
Def.~\ref{def:tc-fun} for sequence and assignment, the following 
intermediate formula:
\begin{align*}
  & \Trf{(\cswI, \cswII)}{(\mem_1, \mem_2)}{\emptyset}{\rela{Q}} =\\
  & \quad\Tcn(\cswII,\mem_2,\emptyset,\\
  & \quad\quad\lambda\mem'_2.
  \forall \mem_3,\mem_5,\mem_7.\\
  & \quad\quad\quad
  \mem_3 = \mem_1[3/\mem_1(\mem_1(1))]\Rightarrow\\
  & \quad\quad\quad
  \mem_5 = \mem_3[\mem_3(1)/\mem_3(\mem_3(2))]\Rightarrow\\
  & \quad\quad\quad
  \mem_7 = \mem_5[\mem_5(2)/\mem_5(3)]\Rightarrow \rassert{Q}{\mem_7,\mem'_2}.
\end{align*}
We can then do the same with $\cswII$ to obtain the final formula:
\begin{align*}
  & \Trf{(\cswI, \cswII)}{(\mem_1, \mem_2)}{\emptyset}{\rela{Q}} =\\
  & \quad\quad\quad
  \forall \sequence{\mem}{8}.\\
  & \quad\quad\quad\quad
  \mem_4 = \mem_2[\mem_2(1)/\mem_2(\mem_2(1))+\mem_2(\mem_2(2))]\Rightarrow\\
  & \quad\quad\quad\quad
  \mem_6 = \mem_4[\mem_4(2)/\mem_4(\mem_4(1))-\mem_4(\mem_4(2))]\Rightarrow\\
  & \quad\quad\quad\quad
  \mem_8 = \mem_6[\mem_6(1)/\mem_6(\mem_6(1))-\mem_6(\mem_6(2))]\Rightarrow\\
  & \quad\quad\quad\quad
  \mem_3 = \mem_1[3/\mem_1(\mem_1(1))]\Rightarrow\\
  & \quad\quad\quad\quad
  \mem_5 = \mem_3[\mem_3(1)/\mem_3(\mem_3(2))]\Rightarrow\\
  & \quad\quad\quad\quad
  \mem_7 = \mem_5[\mem_5(2)/\mem_5(3)]\Rightarrow \rassert{Q}{\mem_7,\mem_8}.
\end{align*}
Here, $\mem_k$ with odd (resp., even) indices result from $\Tcn$ for $\cswI$ (resp., $\cswII$). 
\qed
\end{example}

\smallskip
We similarly define a notation for the %conjunction of 
auxiliary verification conditions
for a sequence of $\nat$ commands.
\begin{definition}[Function $\Tar$]
  \label{def:Tarf}
%  $\Tar$ takes as parameters 
  Given a sequence of commands
  $\sequence{\com}{\nat}$ and a sequence of memory states $\sequence{\mem}{\nat}$, we define function $\Tar$ as follows:
%  and returns a conjunction of formulas.
    \begin{equation*}
      \Tarf{\sequence{\com}{\nat}}{\sequence{\mem}{\nat}}{\memann} \triangleq
      \bigwedge_{i=1}^n \Tanf{\com_i}{\mem_i}{\memann}.
    \end{equation*}
\end{definition}

\begin{remark}
  \label{rmk:Tarf}
%  For a sequence of commands
%  $\sequence{\com}{\nat}$ and a sequence of memory states
%  $\sequence{\mem}{\nat}$, we deduce 
  For $\nat>0,$ it trivially follows from Def.~\ref{def:Tarf} that:
  \begin{equation*}
    \Tarf{\sequence{\com}{\nat}}{\sequence{\mem}{\nat}}{\memann}
    \equiv
    \Tanf{\com_\nat}{\mem_\nat}{\memann} \land
    \Tarf{\sequence{\com}{\nat-1}}{\sequence{\mem}{\nat-1}}{\memann}.
  \end{equation*}
\end{remark}

Using functions $\Tr$ and $\Tar$, we can now give the main result
of this paper:
it states that the verification of relational
properties using the VCGen is correct.
\begin{theorem}[{Soundness of relational VCGen}]
  \label{rela:proof}
  For any sequence of commands $\sequence\com\nat$,
  contract environment $\memann$, procedure environment $\memcom$,
  and relational assertions over $\nat$ states $\rela{P}$ and $\rela{Q}$,
  if the following three properties hold:
  \begin{equation}
    \label{hyp:1}
    \Tfnf{\memann}{\memcom},
  \end{equation}
%  and
  \begin{equation}
    \label{hyp:2}
    \forall \sequence\mem\nat,
    \rassert{P}{\sequence\mem\nat} \Rightarrow
    \Tarf{\sequence\com\nat}{\sequence\mem\nat}{\memann},
  \end{equation}
%  and
  \begin{equation}
    \label{hyp:3}
    \forall \sequence\mem\nat,
    \rassert{P}{\sequence\mem\nat} \Rightarrow
    \Trf{\sequence\com\nat}{\sequence\mem\nat}{\memann}{\rela{Q}},
  \end{equation}
  then we have \enskip
  $%\begin{equation*}
    \rhoared{P}{\sequence\com\nat}{Q}{\memcom}.
  $%\end{equation*}
\end{theorem}

In other words, a relational property is valid
if all procedure contracts are valid, and, assuming the relational
precondition holds, both the %conjunction of 
auxiliary verification conditions
and the main relational verification condition hold.
We give the main steps of the proof below.
The
corresponding 
\coq formalization is available in file \texttt{Rela.v},
and the \coq proof of Theorem~\ref{rela:proof} is in file \texttt{Correct_Rela.v}.

\begin{proof}
%	\commentNK{relecture NK en cours}
  By induction on the length $\nat$ of the sequence of commands $\sequence\com\nat$.
  \begin{itemize}
  \item Induction basis: $\nat = 0$.
    By Def.~\ref{def:rela}, our goal becomes:
    \begin{equation*}
      \rhoared{P}{\sequence\com0}{Q}{\memcom} \,\,\equiv\,\,
      \rassert{P}{[\ ]} \Rightarrow \rassert{Q}{[\ ]}.
    \end{equation*}
    Indeed, by definition of $\Tr$ and Hypothesis (\ref{hyp:3}),
    $\rassert{P}{[\ ]} \Rightarrow \rassert{Q}{[\ ]}$ holds.
  \item Induction step: assuming the result for $\nat$, we prove it for $\nat + 1$. So, assume we have a sequence of commands  $\sequence\com{\nat +1}$, 
  relational assertions  and  environments 
  respecting  (\ref{hyp:1}), (\ref{hyp:2}), (\ref{hyp:3}) (stated for sequences of $\nat + 1$ elements). We have to prove 
  $\hoared{\hat{P}}{\sequence\com{\nat+1}}{\hat{Q}}{\memcom}$, 
  which, by Def.~\ref{def:rela}, is equivalent to:  
  \begin{equation}
  \label{maingoal}
  \forall \sequence{\memvar}{\nat+1}, \sequence{\memvar'}{\nat+1}.\
  \rassert{P}{\sequence{\memvar}{\nat+1}} 
  %    \Rightarrow
  \land
  (\evalrr{\com}{\memvar}{\memvar'}{\psi}{\nat+1})
  \Rightarrow
  \rassert{Q}{\sequence{\memvar'}{\nat+1}}.
  \end{equation}

%  We have to prove
%    \begin{equation}
%      \label{maingoal}
%      \hoared{\hat{P}}{\sequence\com{\nat+1}}{\hat{Q}}{\memcom}.
%    \end{equation}

    First, we can deduce from Hypothesis (\ref{hyp:2}) and Remark~\ref{rmk:Tarf}:
    \begin{equation}
      \label{hyp:13}
      \forall \sequence\mem{\nat+1},\,\rassert{P}{\sequence{\mem}{\nat+1}}
      \Rightarrow \Tanf{\com_{\nat+1}}{\mem_{\nat+1}}{\memann},
    \end{equation}
%    which, according to Remark \ref{rmk:Tarf}, can be rewritten as:
    \begin{equation}
      \label{hyp:17}
      \forall \sequence\mem{\nat+1},\,\rassert{P}{\sequence{\mem}{\nat+1}}
      \Rightarrow
      \Tarf{\sequence{\com}{\nat}}{\sequence{\mem}{\nat}}{\memann}.
    \end{equation}

    By Hypothesis (\ref{hyp:3}) and Def.~\ref{def:Trf}, we have
    \begin{multline}
      \label{hyp:14}
      \qquad \qquad \qquad \qquad \qquad
       \forall \sequence\mem{\nat+1},\,\rassert{P}{\sequence{\mem}{\nat+1}} \Rightarrow \\
      \Tcnf{\com_{\nat+1}}{\mem_{\nat+1}}{\memann}
      {\lambda \mem'_{\nat+1}.
        \Trf{\sequence{\com}{\nat}}{\sequence{\mem}{\nat}}{\memann}
        {\lambda \sequence{\mem'}{\nat}. \rassert{Q}{\sequence{\mem}{\nat+1}}}}.
    \end{multline}

    Using (\ref{hyp:1}), (\ref{hyp:13}) and (\ref{hyp:14}), we can now apply Theorem ~\ref{hoare:proof} (for an arbitrary subsequence $\sequence{\memvar}{\nat}$, 
    %and $\sequence{\memvar'}{\nat}$, 
    that we can thus put in an external universal quantifier)
    to obtain:
    
    \begin{adjustbox}{width=0.96\textwidth}
    \begin{minipage}{130mm}
    \hspace{-10mm}	
    \begin{multline}
    %\footnotesize
    \label{hyp:15}
    \qquad\qquad\qquad\qquad\qquad\qquad\forall \sequence{\memvar}{\nat}.\\ %\sequence{\memvar'}{\nat}
    \hspace{-9mm}
    \hoared{\lambda \mem_{\nat+1}.\, \rassert{P}{\sequence{\mem}{\nat+1}}}{\com_{\nat+1}} 
    {\lambda \mem'_{\nat+1}.\, \Trf{\sequence{\com}{\nat}}{\sequence{\mem}{\nat}}{\memann}
    	{\lambda \sequence{\mem'}{\nat}. \rassert{Q}{\sequence{\mem'}{\nat+1}}}}
    {\memcom}.
    \end{multline}
    \end{minipage}
    \end{adjustbox}
    \medskip

    Using Remark \ref{rmk:Trf} and by rearranging the quantifiers and implications, we can rewrite (\ref{hyp:15}) into:
    \begin{multline}
      \label{hyp:16}
      \forall \sigma_{\nat+1}, \sigma_{\nat+1}'.
      \Vdash\eval{\com_{\nat+1}}{\mem_{\nat+1}}{\memcom}{\mem'_{\nat+1}} \Rightarrow\\
      \forall \sequence{\sigma}{\nat}.
      \rassert{P}{\sequence{\mem}{\nat+1}}
      \Rightarrow
      \Trf{\sequence{\com}{\nat}}{\sequence{\mem}{\nat}}{\memann}
      {\lambda \sequence{\mem'}{\nat}. \rassert{Q}{\sequence{\mem'}{\nat+1}}}.
    \end{multline}

    For arbitrary states
    $\sigma_{\nat+1}$ and $\sigma_{\nat+1}'$ such that
    $\Vdash\eval{\com_{\nat+1}}{\mem_{\nat+1}}{\memcom}{\mem'_{\nat+1}}$, 
    using (\ref{hyp:1}), (\ref{hyp:17}) and (\ref{hyp:16}), we can apply the induction hypothesis, and obtain:
    
\vspace{-5mm}    
    \begin{multline*}
      \label{goal:1}
      \forall \sigma_{\nat+1}, \sigma_{\nat+1}'.
      \Vdash\eval{\com_{\nat+1}}{\mem_{\nat+1}}{\memcom}{\mem'_{\nat+1}} \Rightarrow\\
      \hoared{\lambda \sequence{\mem}{\nat}. \rassert{P}{\sequence{\mem}{\nat+1}}}
      {\sequence{\com}{\nat}}
      {\lambda \sequence{\mem'}{\nat}. \rassert{Q}{\sequence{\mem'}{\nat+1}}}
      {\memcom}.
    \end{multline*}
    Finally, by Def.~\ref{def:rela} and by rearranging the quantifiers, we deduce (\ref{maingoal}).   \qed
  \end{itemize}
\end{proof}

\begin{example}
\label{ex:rel-prop-proof}
%Recall that such formulas (long for a human), 
%are well-treated by automatic solvers.
The relational property of Ex.~\ref{ex:rel-prop} is proven valid 
using the proposed technique based on Theorem~\ref{the:vcghoareproc} in file {\tt Examples.v} 
of the \coq development.
For instance, (\ref{hyp:3}) becomes 
$\forall \mem_1,\mem_2.\,
\rassert{P}{\mem_1,\mem_2} \Rightarrow\Trf{(\cswI, \cswII)}{(\mem_1, \mem_2)}{\emptyset}{\rela{Q}}$,
where the last expression was computed in Ex.~\ref{ex:rel-prop-vc}.
Such formulas---long for a manual proof---are well-treated by automatic solvers.

Notice that in this example we do not need any code transformations or extra separation hypotheses 
in addition to those
anyway needed for the swap functions %to work
%(e.g. the precondition that $\locval_2$ must not point to  $\locval_1,\locval_2,\locval_3$)
while both programs manipulate the same locations $\locval_1,\locval_2$,
and---even worse---the unknown locations pointed by them can be any locations
$\locval_i$, $i>3$.
%The proposed method does not require additional frame rules for procedures
%to ensure that locations are left unchanged by other commands involved
%in the relational property.
\qed
\end{example}

%% file: related.tex
\paragraph{Relational Property Verification.}
%\vspace{-2mm}

Significant work has been done on relational program verification 
%of relational properties 
(see~\cite{47rela,next700} for a detailed
state of the art). We discuss below some of the efforts the most closely related to our work.
%Below we discuss the of the works we consider the most notable.

% \commentLB{
%   A subclass of relational properties, \emph{metamorphic properties},
%   relating multiple executions
%   of the same function~\cite{HuiHuang2013}, are used in the context software testing
%   in order to address the oracle problem~\cite{Weyuker82}.
% }

Various relational logics have been designed as extensions
to Hoare Logic, such as Relational Hoare Logic~\cite{DBLP:conf/popl/Benton04}
and Cartesian Hoare Logic~\cite{SousaD16}.
As our approach, those logics consider
for each command a set of associated memory states
in the very rules of the system,
thus avoiding additional separation assumptions.
Limitations of these logics are often
the absence of support for aliasing
or a limited form of relational properties. For instance,
Relational Hoare Logic supports only relational properties with two commands
and Cartesian Hoare Logic supports only $k$-safety properties (relational properties on
the same command).
Our method has an advanced support of aliasing and
supports a very general definition of relational properties,
possibly between several dissimilar commands.

%We have mention earlier that
Self-compositon~\cite{DBLP:journals/mscs/BartheDR11,DBLP:conf/fm/SchebenS14,blatterKGP17}
and its derivations~\cite{DBLP:conf/fm/BartheCK11,ShemerCAV2021,DBLP:conf/esop/EilersMH18}
are well-known appro\-aches to deal with relational properties.
This is in particular due to their flexibility: self-composition
methods can be applied as a preprocessing step to different verification approaches.
For example, self-composition is used in combination with symbolic execution and model checking
for verification of voting functions~\cite{BeckertBormerEA2016}.
Other examples are the use of self-composition in combination with verification condition generation
in the context of the Java language~\cite{GuillaumeCAD20} or the C language~\cite{blatterKGP17,blatterKGPP18}.
In general, the support of aliasing of C programs
in these last efforts is very limited due the problems mentioned earlier.
Compared to these techniques,
where self-composition is applied before the generation of verification
conditions (and therefore requires taking care about separation of memory
states of the considered programs),
our method can be seen as relating the considered programs' semantics
directly at the level of the verification conditions,
where separation of their memory states is already ensured,
thus avoiding the need to take care of this separation explicitly.

%
% Our method can be seen as a direct application of
%self-composition at the level of
%the verification condition generator avoding taking care of memory separation.

Finally, another advanced approach for relational verification is the
translation of the relational problem into Horn clauses and
their proof using constraint solving%
%Notable work are presented in
~\cite{DBLP:journals/jar/KieferKU18,HiroshCAV2021}.
The benefit of constraint solving lies in the ability to
automatically find relational invariants
and complex self-composition derivations.
% \commentNK{unclear}
Moreover, the translation
of programs into Horn clauses,
done by tools
like \reve\footnote{\url{https://formal.kastel.kit.edu/projects/improve/reve/}},
results in
formulas similar to those generated by our
VCGen.
Therefore, like our approach,
relational verification with constraint solving
requires no additional separation hypothesis in presence of aliasing.

% However, most generation of the formula uses a variable renaming approach
% as postprocession to ensures memory separation. Our approach does
% not require such renaming, since the renaming comes for free from
%% the composition of the verication condition generation.
\vspace{-2mm}
\paragraph{Certified Verification Condition Generation.}
%\vspace{-2mm}
In a broad sense, this work continues previous
efforts
in formalization and
mechanized proof of program language semantics, analyzers and
compilers, %for example,
such as~\cite{SF,leroy-formal-2008,herms-certification-2013,BeringerA19,JourdanLBLP15,JungKJBBD18,WilsJacobsCertCProg2021,KrebbersLW14,BlazyMP15,ParthasarathyMuellerSummers21}.
%\commentNK{Cite:
%
%- SF Pierce et al
%
%- CompCert
%
%- Robert Krebbers Iris?
%
%- DeepSpec Appel \url{https://vst.cs.princeton.edu/}
%
%- André Maroneze, Blazy et al. sur VERASCO ? sur slicing ?  % VMCAI ?
%
%- VeriFast \url{https://arxiv.org/pdf/2110.11034.pdf} ?
%
%}
%\commentLBc{Cite:
%
%- Viper Roots \url{https://www.pm.inf.ethz.ch/research/roots.html}
%}
Generation of certificates (in Isabelle)
for the \textsc{Boogie} verifier is presented 
in~\cite{ParthasarathyMuellerSummers21}.
%A posteriori 
%Regarding certified verification condition generators~\cite{herms-certification-2013,ParthasarathyMuellerSummers21},
%has been studied in the past.
%For example,
The certified deductive
verification tool WhyCert~\cite{herms-certification-2013}
comes with a similar soundness result
for its verification condition generator.
Its formalization follows an alternative proof approach,
based on co-induction, while our proof relies  on induction.
WhyCert is syntactically closer to the C language and the
\acsl specification language~\cite{ACSL},
while our proof uses a simplified
language, but with a richer aliasing model.
%Moreover, 
%%our formalization provides
%%an extra proof  
%we prove (in file {\tt Vcg_Opt.v}) a VCGen 
%optimization (not detailed here for lack of space), % and simplicity.
%which 
%%One of the optimizations~\cite{DBLP:conf/popl/FlanaganS01}
%prevents the size of the generated formulas
%from becoming exponential in the number of
%conditions in the program~\cite{DBLP:conf/popl/FlanaganS01}, which is a classical problem
%for naive verification condition generators.
Furthermore, we provide a formalization and a soundness
proof for relational verification, which was not considered
in WhyCert or in~\cite{ParthasarathyMuellerSummers21}.

%Thus,
To the best of our knowledge, the present work
is the first proposal of
relational property verification based on verification
condition generation realized
for a representative language with procedure calls and aliases
with a full mechanized formalization and proof of soundness in \coq.

%% file: conclusion.tex
We have presented in this paper
%a formalization for
a method for verifying relational properties using a verification
condition generator, without relying on code transformations (such as
self-composition) or making additional
separation hypotheses in case of aliasing.
This method has been fully formalized in \coq, and the soundness of recursive
Hoare triple verification using a verification condition generator (itself formally proved correct) for a simple language with procedure calls and aliasing
has been formally established.
%Overall,
Our formalization is well-adapted for proving
possible optimizations of a VCGen
and for using optimized VCGen versions for relational property verification.

%We expect this work to set up a basis
%to be a first step in
This work sets up a basis
for the formalization
of modular verification of relational properties using
verification condition generation. 
We plan to extend it with more features such as the possibility to refer to the values 
of variables before a function call in the postcondition (in order to relate them to the values after the call) and the capacity to rely on relational properties
during the proof of other properties.
Future work also includes 
an implementation of this technique inside a tool like RPP~\cite{blatterKGP17}
in order to integrate it with SMT solvers and to evaluate it on benchmarks.
The final objective would be
to obtain a system similar to the verification of Hoare triples,
namely, having relational procedure contracts, relational assertions,
and relational invariants.
Currently, for relational properties, product programs~\cite{DBLP:conf/fm/BartheCK11}
or other self-composition optimizations~\cite{ShemerCAV2021}
are the standard approach to deal with complex loop constructions.
We expect that user-provided coupling invariants and
loop properties can avoid having to rely on code transformation methods.
Moreover, we expect termination and co-termination
\cite{DBLP:conf/cade/HawblitzelKLR13},\cite{HiroshCAV2021} to be used to
extend the modularity of relational contracts.

%% file: appendix.tex
\noindent
%This appendix is provided for convenience of the res, not for publication.
%For convenience of the reviewers,

%\commentNK{in progress}

\section{Detailed Motivating Example}
\label{app-motiv-ex}

Figure \ref{fig:ex-rel-prop-swaps-self-full} provides a more detailed version 
of the motivating example presented in Section \ref{sec:intro} and Fig.~\ref{fig:ex-rel-prop-swaps-self}.
Programs $\CprogI$ and $\CprogII$ contain, resp., C functions 
\lstinline'sw1' and \lstinline'sw2', where pointers 
\lstinline'x1' and \lstinline'x2' are function parameters and variable
\lstinline'x3' in $\CprogI$ is a local variable. 
This choice is most natural in C.

Recall that we consider a relational property $\Rsw$: % stating that 

\begin{itemize}
\item[$\Rsw$:]
both programs,
executed from two states 
%NK: I propose this statement to avoid a misunderstanding that  *x1 == *x2 must hold for each program
in which \lstinline'*x1' has the same value for both programs
and \lstinline'*x2' has the same value for both programs, 
will end up in two states in which each of these locations
also has the same value. 	
\end{itemize}

To prove the target relational property, 
a new composed C program $\CprogIII$ with a C function
\lstinline'sw3' is created 
(see Fig.~\ref{fig:ex-rel-prop-swaps-self-full})
by composing the code of both functions.
To distinguish variables of different programs, variables coming from  $\CprogI$ 
and $\CprogI$ are marked, resp., with a suffix ``\lstinline'_1'''
or ``\lstinline'_2'''.

In the self-composition based approach,
the target relational property for the composed program $\CprogIII$ 
is proved by the Hoare triple  $\left\{{P}\right\}\CprogIII\left\{{Q}\right\}$,
where precondition $P$ and postcondition $Q$ are defined in 
Fig.~\ref{fig:ex-rel-prop-swaps-self-full}.
The definitions are expressed in the \acsl specification language~\cite{ACSL}.
Lines 5--6 in the definition of $P$ 
state that each of \lstinline'*x1' and \lstinline'*x2' has the same value
in the states before the execution of \lstinline'sw1' and \lstinline'sw2'. 
Similarly, lines 5--6 in the definition of $Q$ 
state the same properties
%that each of \lstinline'*x1' and \lstinline'*x2' has the same value
after the execution of \lstinline'sw1' and \lstinline'sw2'.
However, the precondition must also include additional constraints.
Lines 9--11 in the definition of $P$ provide usual preconditions
for the swap function \lstinline'sw1' to be executed correctly:
the input pointers must be valid and separated.
For instance, validity of pointer \lstinline'x1_1' means that  \lstinline'*x1_1' can be safely read and written. 
The separation property \lstinline'\separated(x1_1,x2_1)' means that the locations 
\lstinline'*x1_1' and \lstinline'*x2_1' are disjoint, 
that is, do not share any byte\footnote{Notice that 
this separation property is stronger in C than the non-equality constraint \lstinline'x1_1 != x2_1', which does not exclude that both locations have some bytes in common (if the pointers are not aligned). 
For simplicity, byte-related data representation and alignment 
constraints are not modeled in $\mathcal{L}$, where the separation can be
simply represented by non-equality constraints. This does not  
restrict the representativity of $\mathcal{L}$ for the purpose of our study.}.
Lines 14--16 in the definition of $P$ provide similar preconditions
for the swap function \lstinline'sw2'.
%\footnote{
For 
simplicity, we ignore arithmetic overflows in \lstinline'sw2':
the specification and 
verification of properties about the absence of arithmetic overflows
are straightforward and orthogonal to the purpose of this paper.

Notice that thanks to the choice of having pointers 
\lstinline'x1' and \lstinline'x2' as function parameters and variable
\lstinline'x3' in $\CprogI$ as a local variable,
for this version 
we do not need to state explicitly other separation hypotheses
stating that  \lstinline'x1' and \lstinline'x2'  do not refer to 
\lstinline'x1', \lstinline'x2' and, for \lstinline'sw1', \lstinline'x3' themselves.
Indeed, these separation hypotheses\footnote{In $\mathcal{L}$, for simplicity, 
we consider only global variables, therefore, in the 
counterparts $\cswI$ and $\cswII$ in language $\mathcal{L}$, these 
additional separation hypotheses must be explicit (as we show in  
Ex.~\ref{ex:rel-prop} and Fig.~\ref{fig:ex-rel-prop-swaps}).
This slight difference of modeling is intentional in order to show the most
natural version of these functions in C with function parameters and local 
variables rather than with global variables only.} 
are already ensured by the fact
that  \lstinline'x1', \lstinline'x2' and, for \lstinline'sw1', \lstinline'x3'  are
allocated during the call to the C function
(and the pointers  \lstinline'x1' and \lstinline'x2'  are 
valid before the call).
 
%, for instance,
%\lstinline'\separated(x1_1,&x1_1)'.
\bigskip
The aforementioned parts of $P$ and $Q$ naturally come 
from the relational property \Rsw{} and the preconditions 
of the considered functions: in this sense, they are expected. 
However, they are not sufficient: 
in a real-life language with possible aliasing like C, 
to model the behavior of both programs correctly within the composed program 
and to prove the expected relational property, additional separation 
hypotheses between the variables coming from both programs
$\CprogI$ and $\CprogII$ are required.
They are expressed by lines 20--23 in the definition of $P$ in
Fig.~\ref{fig:ex-rel-prop-swaps-self-full}.

Such additional separation hypotheses become even more complex for
real-life programs, in particular in C, 
with a greater number of pointers and/or 
in the presence of multiple pointers (such as double 
pointers, for instance, \lstinline'int **p'). 
Indeed, the required separation hypotheses 
for the composed program rapidly become extremely hard 
to specify (or to generate)
%completely 
in order to ensure a sound proof of relational properties 
on the composed program. 

\bigskip
With this definition of  precondition $P$ and postcondition $Q$, the 
code of $\CprogIII$ can be proved to satisfy its contract by 
the deductive verification plugin \WP of
\framac~\cite{Frama-C}.

\begin{figure}[h]
\normalfont
\centering
\begin{tabular}{c|rcl}

\begin{minipage}{5.55cm}
\begin{lstlisting}[stepnumber=0,mathescape]
//$\mbox{\rm C program }\CprogI:$
void sw1(int *x1,int *x2){
  int x3;
  x3  = *x1;
  *x1 = *x2;
  *x2 = x3; 
}
$ $
\end{lstlisting}
\begin{lstlisting}[stepnumber=0,mathescape]
//$\mbox{\rm C program }\CprogII:$
void sw2(int *x1,int *x2){
  *x1 = *x1 + *x2;
  *x2 = *x1 - *x2;
  *x1 = *x1 - *x2;
}  
\end{lstlisting}	
\end{minipage}

%\enskip&\enskip
&

$\left\{{P}\right\}$

&%\enskip
		
\begin{minipage}{5.8cm}
\begin{lstlisting}[stepnumber=0,mathescape]
//$\mbox{\rm Composed C program }\CprogIII:$
void sw3(int *x1_1,int *x2_1
    int *x1_2,int *x2_2){

//$\mbox{\rm Code simulating }\CprogI:$
  int x3_1;
  x3_1  = *x1_1;
  *x1_1 = *x2_1;
  *x2_1 = x3_1;

//$\mbox{\rm Code simulating }\CprogII:$
  *x1_2 = *x1_2 + *x2_2;
  *x2_2 = *x1_2 - *x2_2;
  *x1_2 = *x1_2 - *x2_2;
}
\end{lstlisting}	
\end{minipage}

&%\enskip

$\left\{{Q}\right\}$

\end{tabular}
\\[3mm]
\begin{tabular}{c|rcl}
	
\begin{minipage}{5cm}
\begin{lstlisting}[mathescape]
//$P\mbox{ \rm is defined as follows:}$

//$\mbox{\rm Relation between initial }$
//$\mbox{\rm values of }\CprogI \mbox{ \rm and } \CprogII:$
*x1_1 == *x1_2 &&
*x2_1 == *x2_2 &&

//$\mbox{\rm Preconditions for }\CprogI:$
\valid(x1_1) &&
\valid(x2_1) &&
\separated(x1_1,x2_1) &&

//$\mbox{\rm Preconditions for }\CprogII:$
\valid(x1_2) &&
\valid(x2_2) &&
\separated(x1_2,x2_2) &&

//$\mbox{\rm Extra hypotheses for }$
//$\mbox{\rm a correct simulation by }\CprogIII:$
\separated(x1_1,x1_2) &&
\separated(x1_1,x2_2) &&
\separated(x2_1,x1_2) &&
\separated(x2_1,x2_2)
\end{lstlisting}	
%\separated(x1_1,&x3)  &&
%\separated(x2_1,&x3)  &&
\end{minipage}
	
\qquad&\qquad	
\begin{minipage}{5cm}
\begin{lstlisting}[mathescape]
//$Q\mbox{ \rm is defined as follows:}$

//$\mbox{\rm Relation between resulting }$
//$\mbox{\rm values of }\CprogI \mbox{ \rm and } \CprogII:$
*x1_1 == *x1_2 &&
*x1_1 == *x1_2

$\mbox{ }$
\end{lstlisting}	
\end{minipage}

\end{tabular}
\vspace{-3mm}
\caption{Two C programs $\CprogI$ and $\CprogII$ swapping \lstinline'*x1' and 
\lstinline'*x2' and the Hoare triple 
$\left\{{P}\right\}\CprogIII\left\{{Q}\right\}$
to prove a relational property between them using 
their composition in C program $\CprogIII$, as well as 
%(simplified) 
definitions of precondition
$P$ and postcondition $Q$ of $\CprogIII$.}
\vspace{-5mm}
\label{fig:ex-rel-prop-swaps-self-full}
\end{figure}

\clearpage
\section{Complete Semantics of Language $\mathcal{L}$}
\label{app-sem}

\subsection{Evaluation of Arithmetic and Boolean Expressions in $\mathcal{L}$}
\label{app-sem-exp-eval}
%\bigskip

We provide
a complete list of rules for evaluation of arithmetic and Boolean expressions 
in $\mathcal{L}$ in Fig.~\ref{fig:exp-eval-full}.
Evaluation of arithmetic and Boolean expressions in $\mathcal L$ is defined by functions
$\EAexp$ and $\EBexp$.
As mentioned above, the subtraction is lower-bounded by 0.
Operations $*\locval_i$ and $\&\locval_i$ have a
semantics similar to the C language, {i.e.} dereferencing and address-of.
Semantics of Boolean expressions is standard~\cite{DBLP:books/daglib/0070910}.

%a VirtualBox virtual machine (under
%Ubuntu Linux) that includes the presented tools and the benchmarks. It allows to
%run the tools and observe the results.
%
%The VM is (or will be soon) available at
%
%\begin{center}
%	\url{  https://bit.ly/3psdKXy } \\
% 	%leads to \url{ ... }
% 	login: user\\
% 	password: user
%\end{center}

\begin{figure}[tb]
  \centering
  \begin{minipage}{6cm}
    \begin{align*}
      \EAexpf{\nat}{\mem} & \triangleq \nat\\
      \EAexpf{\locval_i}{\mem} & \triangleq \mem(i)\\
      \EAexpf{*\locval_i}{\mem} & \triangleq \mem(\mem(i))\\
      \EAexpf{\&\locval_i}{\mem} & \triangleq i\\
      \EAexpf{\aexp_1 \opa \aexp_2}{\mem} & \triangleq
                                            \EAexpf{\aexp_1}{\mem} \opa \EAexpf{\aexp_2}{\mem}
    \end{align*}
  \end{minipage}
  \begin{minipage}{4cm}
    \begin{align*}
      \EBexpf{true}{\mem} & \triangleq \True\\
      \EBexpf{false}{\mem} & \triangleq \False\\
      \EBexpf{\aexp_1 \opb \aexp_2}{\mem} & \triangleq
                                            \EAexpf{\aexp_1}{\mem} \opa \EAexpf{\aexp_2}{\mem}\\
      \EBexpf{\bexp_1 \opl \bexp_2}{\mem} & \triangleq
                                            \EBexpf{\bexp_1}{\mem} \opl \EBexpf{\bexp_2}{\mem}\\
      \EBexpf{\neg\bexp}{\mem} & \triangleq  \neg\EBexpf{\bexp}{\mem}
    \end{align*}
  \end{minipage}
  \vspace{-1mm}
  \caption{Evaluation of arithmetic and Boolean expressions in $\mathcal{L}$.}
  \vspace{-5mm}
  \label{fig:exp-eval-full}
\end{figure}

\subsection{Operational Semantics of Commands in $\mathcal{L}$ in $\mathcal{L}$}
\label{app-sem-sper-sem}

We provide
a complete operational semantics of commands in $\mathcal{L}$ in Fig.~\ref{fig:oper-sem-full}.

\begin{figure}[tb]
  \centering
  \begin{minipage}{3cm}
    \begin{equation*}
      % \tag{\textsc{skip}}
      \inferrule*
      {}
      {\eval{\wskip}{\memvar}{\memcom}{\memvar}}
    \end{equation*}
  \end{minipage}
  \begin{minipage}{4cm}
    \begin{equation*}
      % \tag{\textsc{assign}}
      \inferrule*
      {\EAexpf{\aexp}{\memvar} = \nat}
      {\eval{\locval_i := \aexp}{\memvar}{\memcom}{\memvar[i/n]}}
    \end{equation*}
  \end{minipage}
  \begin{minipage}{4cm}
    \begin{equation*}
      % \tag{\textsc{assign-mem}}
      \inferrule*
      {\EAexpf{\aexp}{\memvar} = \nat}
      {\eval{*\locval_i := \aexp}{\memvar}
        {\memcom}\memvar[\memvar(i)/n]}
    \end{equation*}
  \end{minipage}

  \begin{minipage}{5cm}
    \begin{equation*}
      % \tag{\textsc{assert}}
      \inferrule*
      % {P(\memvar)}
      {}
      {\eval{\wassert{P}}{\memvar}
        {\memcom}{\memvar}}
    \end{equation*}
  \end{minipage}
  \begin{minipage}{6cm}
    \begin{equation*}
      % \tag{\textsc{if true}}
      \inferrule*
      {\EBexpf{\bexp}{\memvar} = \True\\
        \eval{\com_1}{\memvar_1}{\memcom}\memvar_2}
      {\eval{\wif{\bexp}{\com_1}{\com_2}}{\memvar_1}{\memcom}\memvar_2}
    \end{equation*}
  \end{minipage}

  \begin{minipage}{5cm}
    \begin{equation*}
      % \tag{\textsc{sequence}}
      \inferrule*
      {\eval{\com_1}{\memvar_1}{\memcom}\memvar_2\\
        \eval{\com_2}{\memvar_2}{\memcom}\memvar_3}
      {\eval{\com_1;\com_2}{\memvar_1}{\memcom}\memvar_3}
    \end{equation*}
  \end{minipage}
  \begin{minipage}{6cm}
    \begin{equation*}
      % \tag{\textsc{if false}}
      \inferrule*
      {\EBexpf{\bexp}{\memvar} = \False\\
        \eval{\com_2}{\memvar_1}{\memcom}\memvar_2}
      {\eval{\wif{\bexp}{\com_1}{\com_2}}{\memvar_1}{\memcom}\memvar_2}
    \end{equation*}
  \end{minipage}

  \begin{equation*}
    % \tag{\textsc{while true}}
    \inferrule*
    {\EBexpf{\bexp}{\memvar_1} = \True\\
      \eval{\com_1}{\memvar_1}{\memcom}\memvar_2\\
      \eval{\wwhilei{\bexp}{P}{\com}}{\memvar_2}{\memcom}\memvar_3}
    {\eval{\wwhilei{\bexp}{P}{\com}}{\memvar_1}{\memcom}\memvar_3}
  \end{equation*}

  \begin{minipage}{6cm}
    \begin{equation*}
      % \tag{\textsc{while false}}
      \inferrule*
      {\EBexpf{\bexp}{\memvar} = \False}
      {\eval{\wwhilei{\bexp}{P}{\com}}{\memvar}{\memcom}\memvar}
    \end{equation*}
  \end{minipage}
  \begin{minipage}{5.5cm}
    \begin{equation*}
      % \tag{\textsc{call}}
      \inferrule*
      {\eval{\body{\loccom}{\memcom}}{\memvar_1}{\memcom}\memvar_2}
      {\eval{\wcall{\loccom}}{\memvar_1}{\memcom}\memvar_2}
    \end{equation*}
  \end{minipage}
  \vspace{-1mm}
  \caption{Operational semantics of commands in $\mathcal{L}$.}
  \vspace{-5mm}
  \label{fig:oper-sem-full}
\end{figure}

%\dots
%In case of a temporary access problem, the reviewers are kindly asked to contact the authors via
%the PC chairs.
%\dots

% \paragraph{Disclaimer.} Due to virtualization, the execution time of the tools
% on the VM can be different from the time observed by the authors without
% virtualization and can depend on the resources allocated to the VM.